\documentclass[journal,a4paper]{IEEEtran}

\usepackage{amsmath}
\usepackage{amsfonts}
\usepackage{amsthm}
\usepackage{graphicx}
\usepackage{cite}
\usepackage{bm}
\newcommand{\ma}[1]{\mathbf{ #1 }}         

\newcommand{\compl}{\mathbb{C}}        
\newcommand{\real}{\mathbb{R}}         

\usepackage{color}

\usepackage{mathtools}
\usepackage{geometry}
\newcommand{\ignore}[1]{}

\ifCLASSINFOpdf

\else

\fi

\usepackage{enumerate}

\usepackage{algorithm}
\usepackage{algpseudocode}

\usepackage{multicol}
\usepackage{float}
\usepackage{graphicx}

\usepackage{url}
\usepackage[utf8]{inputenc}
\usepackage[english]{babel}

\newtheorem{theorem}{Theorem}[section]
\newtheorem{lemma}[theorem]{Lemma}

\begin{document}

%

\title{Sum Secrecy Rate Maximization in a Multi-Carrier MIMO Wiretap Channel with Full-Duplex Jamming}
\author{\IEEEauthorblockN{ {Tianyu~Yang,} {Omid~Taghizadeh},  {and} {Rudolf Mathar}}

\IEEEauthorblockA{Institute for Theoretical Information Technology,~RWTH Aachen University,~D-52074 Aachen, Germany \\
Email: \{yang, taghizadeh, mathar\}@ti.rwth-aachen.de}
}

\maketitle

\begin{abstract}
In this paper we address a sum secrecy rate maximization problem for a multi-carrier and MIMO communication system. We consider the case that the receiver is capable of full-duplex (FD) operation and simultaneously sends jamming signal to a potential eavesdropper. In particular, we simultaneously take advantage of the spatial and frequency diversity in the system in order to obtain a higher level of security in the physical layer. Due to the non-convex nature of the resulting mathematical problem, we propose an iterative solution with a guaranteed convergence, based on block coordinate descent method, by re-structuring our problem as a separately convex program. Moreover, for the special case that the transmitter is equipped with a single antenna, an optimal transmit power allocation strategy is obtained analytically, assuming a known jamming strategy. We also study a FD bidirectional secure communication system, where the jamming power can be reused to enhance the sum secrecy rate. The performance of the proposed design is then numerically evaluated compared to the other design strategies, and under different system assumptions.     
{\let\thefootnote\relax\footnote{{Part of this work has been presented in ICC’17-WT07, the 2017 IEEE International Conference on Communications Workshops \cite{7962841}.}}}
\addtocounter{footnote}{-1}\let\thefootnote\svthefootnote
\end{abstract}

\begin{keywords}
Full-duplex, wiretap channel, secrecy rate, jamming, multi-carrier, MIMO.
\end{keywords}



%
\IEEEpeerreviewmaketitle

\section{Introduction} \label{sec:into}
Full-Duplex transceivers are known for their capability to enhance various aspects of wireless communication systems, e.g., achieving higher spectral efficiency and physical layer security, due to the simultaneous transmission and reception capability on the same channel~\cite{6736751}. Nevertheless, such systems suffer from the inherent self-interference (SI) from their own transmitter. Recently, specialized cancellation techniques, e.g., \cite{179789,6319352,Bharadia:2013}, have demonstrated an adequate level of isolation between Tx and Rx directions to facilitate a FD communication and motivated a wide range of related studies, see, e.g., \cite{6832464}. 
As an interesting use case of such capability, it is known that a FD receiver can significantly enhance the security of a wireless system in physical layer, by simultaneously transmitting a jamming signal to a potential eavesdropper, while receiving the useful information from the legitimate transmitter. 
Note that the information security of the current communication systems are typically addressed by distributing secret keys, using cryptographic approaches. This approach mainly relies on the assumption that a potential eavesdropper has a limited computational power and hence may not break the exchanged secret key. On the other hand, this assumption is increasingly undermined due the advances in the production of digital processors, and leads to a growing interest to ensure the security of information systems in the physical layer. In this regard, the concept of the wiretap physical channel is introduced in \cite{wyner1975wire}, including a legitimate transmitter, namely Alice, a legitimate receiver, namely Bob, as well as an eavesdropper, namely Eve. In this regard, the secrecy capacity of a wiretap channel is defined as the information capacity that can be exchanged among the legitimate users, without being accessible by Eve. The secrecy capacity of the defined wiretap model have been since extensively studied for various systems, regarding performance bounds, channel coding and information theoretic aspects, as well as the system design and resource optimization, see \cite{6739367} and the references therein. \par

The application of FD capability for secrecy rate maximization in a wiretap channel is studied in \cite{6542749} where a FD Bob is capable of transmitting jamming signal, considering a single antenna Alice and a passive eavesdropper. In particular the utilization of a FD jammer reduces the need to external helpers, which are commonly used to degrade the reception capability of the eavesdropper via cooperative jamming \cite{5352243,5638162}, without having to trust external nodes, or demanding additional resources. The studied system \cite{6542749} is then extended to a setup where all nodes are equipped with multiple antennas \cite{7339654,6787008}. Moreover, extensions on the operation the FD Bob is introduced by considering a simultaneous information and jamming transmission, i.e., operating as a base station \cite{6936336}, and considering a FD jamming Bob that simultaneously relays information to a third node, i.e., operating as a jamming relay \cite{7136146}. The consideration of a joint FD operation of both Alice and Bob, i.e., a bi-directional wiretap channel, as well as the possibility of the FD operation at Eve, i.e., an active eavesdropper, is respectively studied in \cite{7414075, 6189999}, and in \cite{7391133}, targeting at sum-secrecy rate maximization in both communication directions. \par
  
The aforementioned works study different system possibilities considering a single-carrier, frequency-flat channel model for all of the physical links. In contrast, the consideration of a frequency selective, multi-carrier system in the context of sum secrecy rate maximization is extensively studied for a wireless system with half-duplex links, see \cite{5961648, 6187727, 5872025, 6516879}. In this respect, extension of the prior works with FD transceivers to a frequency-selective and multi-carrier design is interesting. This is since the usual flat-fading assumption of the previous studies limits the usability of the proposed designs. Furthermore, in a frequency selective setup, the frequency diversity in different subcarriers can be opportunistically used, both regarding the jamming and the desired information link to jointly enhance the resulting secrecy capacity. In this respect, a power-auction game is proposed in \cite{7386164} for maximizing the sum secrecy rate in a FD and multi-carrier system, where all nodes are equipped with a single antenna. However, such studies are not yet extended for a system with multiple-antenna FD transceivers. \par
In this paper we address a joint power and beam optimization problem for sum secrecy rate maximization in a multi-carrier and MIMO wiretap channel. We consider the case that the receiver is capable of full-duplex (FD) operation and simultaneously sends jamming signal to a potential eavesdropper. In particular, we simultaneously take advantage of the spatial and frequency diversity in the system in order to obtain a higher level of security in the physical layer. In Section~\ref{sec:model}, the system model is presented. In Section~\ref{sec:SSRM} the corresponding optimization strategy is defined. Due to the non-convex nature of the resulting mathematical problem, we propose an iterative solution with a guaranteed convergence, based on block coordinate descent method \cite[Subsection~2.7]{bertsekas1999nonlinear}, by re-structuring our problem as a separately convex program. Moreover, for special case that the transmitter is equipped with a single antenna, an optimal transmit power allocation strategy is obtained analytically, assuming a known jamming strategy. In Section~\ref{sec:extendedsolution} the system is extended to a FD bidirectional communication system. The performance of the proposed design is then numerically evaluated in Section~\ref{sec:simulations}.

\subsection{Mathematical Notation:}
Throughout this paper, column vectors and matrices are denoted as lower-case and {upper-case} bold letters, respectively. {Mathematical expectation, trace}, inverse, determinant, transpose, conjugate {and} Hermitian transpose are denoted by $ \mathbb{E}\{\cdot\}, \; {\text{tr}}(\cdot), \; (\cdot)^{-1}\; |\cdot|, \; (\cdot)^{ T},\; (\cdot)^{*}$ {and} $(\cdot)^{ H},$ respectively. The Kronecker product is denoted by $\otimes$. The identity matrix with dimension $K$ is denoted as ${\ma I}_K$ and ${\text{vec}}(\cdot)$ operator stacks the elements of a matrix into a vector. $\ma{0}_{m \times n}$ represents an all-zero matrix with size $m \times n$. $\bot$ represents the statistical independence. $\text{diag}(\cdot)$ returns a diagonal matrix by putting the off-diagonal elements to zero. The sets of real, non-negative real, complex, natural numbers and the set of all positive semi-definite matrices with Hermitian symmetry are respectively denoted by $\mathbb{R}$, $\mathbb{R}^+$, $\mathbb{C}$, $\mathbb{N}$ and $\mathcal{H}$. $\|\cdot\|_F$ represents Frobenius norm. $\{a\}^{+}$ is equal to $a\in \real$ if $a \geq 0$, and zero otherwise.   


\section{System Model}\label{sec:model}
\begin{figure}[!t] 
    \begin{center}
        \includegraphics[angle=0,width=0.90\columnwidth]{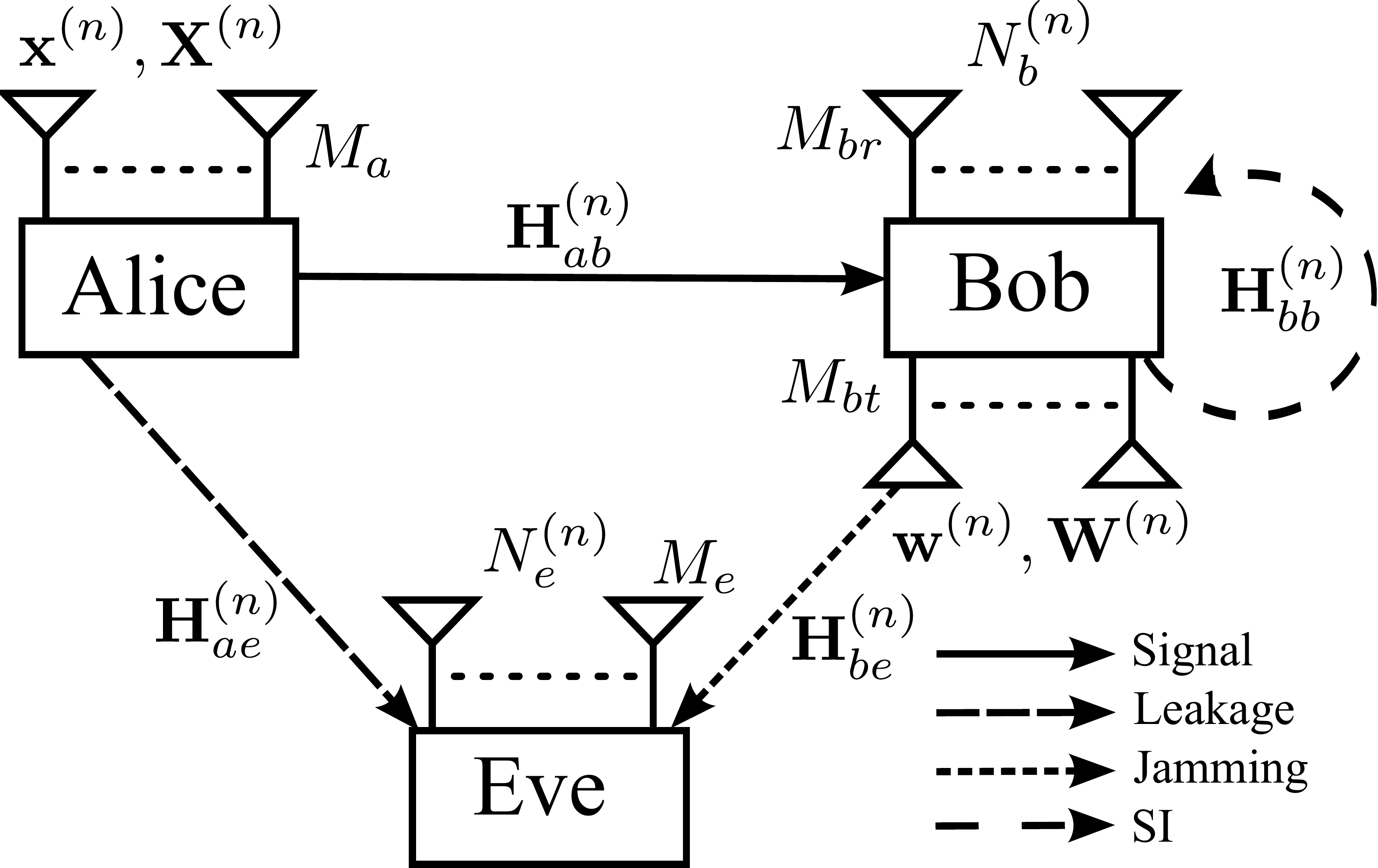}
    \caption{{{The studied multi-carrier wiretap channel, including Alice (legitimate transmitter), Bob (legitimate receiver) and Eve (eavesdropper). Bob is capable of FD operation. Upper-index $n$ is the subcarrier index. }}} \label{fig:sysmodel}
    \end{center} \vspace{-0mm} 
\end{figure}
We consider a classic wiretap channel: Alice transmits a massage to Bob while Eve intends to eavesdrop the transmitted message. Moreover, we consider a multi-carrier and MIMO system, where Bob is capable of FD operation by sending a jamming signal to Eve while receiving the message from Alice, see Fig.~\ref{fig:sysmodel}. Alice and Eve are respectively equipped with $M_{a}$ and $M_{e}$ transmit and receive antennas, where Bob is equipped with $M_{bt}$ ($M_{br}$) transmit (receive) antennas. In each subcarrier, channels are assumed to follow a quasi-stationary and flat-fading model. In this regard, channels from Alice to Bob (desired communication channel), Alice to Eve (information leakage channel), Bob to Bob (SI channel), and Bob to Eve (jamming channel) are respectively denoted as $\ma{H}_{ab}^{(n)}\in \compl^{M_{br}\times M_{a}}$, $\ma{H}_{ae}^{(n)}\in \compl^{M_{e}\times M_{a}}$, $\ma{H}_{bb}^{(n)}\in \compl^{M_{br}\times M_{bt}}$ and $\ma{H}_{be}^{(n)}\in \compl^{M_{e}\times M_{bt}}$, where $n \in \mathcal{N}$, and $\mathcal{N}$ is the index set of all subcarriers. The transmit signal from Alice can be written as  
\begin{align} \label{model_x_alice}
\ma{x}^{(n)} = \ma{V}^{(n)} \ma{s}^{(n)}, 
\end{align}
where $\ma{s}^{(n)} \sim \mathcal{C} \left( \ma{0}_{d\times1} , \ma{I}_d \right)$ and $\ma{V}^{(n)} \in \compl^{M_a \times d}$ respectively represent the vector of data symbols to be transmitted, and the transmit precoder for subcarrier $n$. Moreover, $d\in\mathbb{N}$ is the number of data streams which are transmitted in parallel. On the other hand, the jamming signal transmitted by Bob is described as $\ma{w}^{(n)} \sim \mathcal{CG} \left( \ma{0}_{{M_{bt}}\times1} , \ma{W}^{(n)} \right)$, where $\ma{W}^{(n)} \in \compl^{M_{bt} \times M_{bt}}$ is the jamming transmit covariance, and $\mathcal{G}$ represents Gaussian distribution. Note that the transmitted jamming signal by Bob impacts the wiretap channel in two opposite directions. Firstly, the jamming signal impacts the signal received by Eve as an additional interference term, and degrades the Alice to Eve channel. Secondly, due to the imperfect SI cancellation, the residual interference terms lead to the degradation of the communication channel between Alice and Bob. Since the aforementioned effects impact the achievable system secrecy in opposite directions, a smart design of the jamming strategy is crucial. The received signal by Eve is written as 
\begin{align} \label{eq:model_y_Eve}
\ma{y}_{e}^{(n)} =  \ma{H}_{ae}^{(n)} \ma{x}^{(n)} + \ma{H}_{be}^{(n)}\ma{w}^{(n)} + \ma{n}_{e}^{(n)}, 
\end{align}
where $\ma{n}_{e}^{(n)} \sim \mathcal{CG} \left( \ma{0}_{M_{e}  \times 1 }, N_e^{(n)} \ma{I}_{M_{e}} \right)$ is the additive white noise on Eve. Similarly, the received signal at Bob is formulated as  
\begin{align} \label{eq:model_y_bob}
\ma{y}_{b}^{(n)} =  \ma{H}_{ab}^{(n)} \ma{x}^{(n)}  + \ma{n}_{b}^{(n)} +  \ma{z}_{b}^{(n)}, 
\end{align}
where $\ma{n}_{b}^{(n)} \sim \mathcal{CG} \left( \ma{0}_{M_{br} \times 1}, N_b^{(n)} \ma{I}_{M_{br}} \right)$ is the additive white noise on Bob, and $\ma{z}_{b}^{(n)} \in \compl^{M_{br}}$ is the baseband representation of the residual SI signal in subcarrier $n$, remaining from the SI cancellation process. 

\subsection{Residual SI model}
We recognize three different sources of error considering the state-of-the-art SI cancellation methods~\cite{6177689}. This includes inaccuracy of the channel state information (CSI) regarding SI path as well as the inaccuracy of the transmit/receive chain elements in the analog domain. In the following we study the impact of each part separately. 

\subsubsection{Linear SI cancellation error}
The estimation accuracy of the CSI in the SI path is limited, particularly in the environments with limited channel coherence time, see~\cite[Subsection~3.4.1]{Jain:2011}, \cite[Subsection~V.C]{6404659}. In this respect, the error of the CSI estimation regarding the SI path is defined as $\ma{E}_{bb}^{(n)}$, such that ${\ma{E}}_{bb}^{(n)} =  {\ma{D}}_{bb}^{(n)} \bar{\ma{E}}_{bb}^{(n)} $ where $\bar{\ma{E}}_{bb}^{(n)}$ is matrix of zero-mean i.i.d. elements with unit variance, and ${\ma{D}}_{bb}^{(n)}$ incorporates spatial correlation, see \cite[Equation~(8),~(9)]{6177689}.  

\subsubsection{Transmitter distortion}

Similar to \cite{6177689}, the inaccuracy of the analog (hardware) elements in the transmit chains, e.g., digital-to-analog converter error, power amplifier noise and oscillator phase noise, are jointly modeled by injecting an additive Gaussian distortion signal term for each transmit chain. This is written as $q_l (t) = {e}_{\text{t},l} (t) + w_l (t)$, see Fig.~\ref{fig:TransceiverAccuracyModel} such that 
\begin{align} 
& {e}_{\text{t},l} (t) \sim \mathcal{CG} \left( 0, \kappa \mathbb{E} \big\{  w_l (t) {w_l (t)}^{*} \big\} \right), \; \nonumber \\  & {e}_{\text{t},l} (t) \bot {w}_{l} (t),\; {e}_{\text{t},l} (t) \bot {e}_{\text{t},l} (t{'}), \; {e}_{\text{t},l} (t) \bot  {{e}_{\text{t},{l^{'}}}} (t),  \label{eq_model_distortion_stat_1}
\end{align}
where $w_l$, $e_{\text{t},l}$, and $q_l \in \compl$ respectively represent the intended (distortion-free) transmit signal, additive transmit distortion, and the actual transmit signal from the $l$-th transmit chain, and $t$ denotes the instance of time\footnote{Note that the signal representation in time domain includes the superposition of signal parts in all subcarriers.}. Moreover, we have ${t \neq t^{'}}, {l \neq l^{'}}$, and $\kappa  \in \real^+$ is the distortion coefficient, relating the collective power of the distortion signal to the intended transmit power. \par

\subsubsection{Receiver distortion}

Similar to the transmit chain, the combined effect of the inaccurate hardware elements, e.g., analog-to-digital converter error, oscillator phase noise and automatic gain control noise, are presented as an additive distortion term $\tilde{q}_l (t) = {e}_{\text{r},l} (t) + u_l (t)$ such that 
\begin{align} 
& {e}_{\text{r},l} (t) \sim \mathcal{CG} \left( 0, \beta \mathbb{E} \big\{  u_l (t){u_l (t)}^{*} \big\} \right),  \nonumber \\  & {e}_{\text{r},l} (t) \bot {u}_{l} (t),\; {e}_{\text{r},l} (t) \bot {e}_{\text{r},l} (t^{'}), \; {e}_{\text{r},l} (t) \bot  {{e}_{\text{r},{l^{'}}}} (t), \label{eq_model_distortion_stat_2}
\end{align}
where $u_l$, $e_{\text{r},l}$, and , $\tilde{q}_l$ respectively represent the intended (distortion-free) receive signal, additive receive distortion, and the received signal from the $l$-th receive antenna. Moreover, $\beta \in \mathbb{R}^+$ holds a similar role as $\kappa$ regarding the distortion signal variance in the receiver side. \par

Note that the defined model for transmit/receive distortion terms follow two important intuitions. Firstly, unlike the usual thermal noise model, the variance of the distortion terms are proportional to the transmit/receive signal power in each chain. Secondly, the distortion signal is statistically independent to the intended transmit/receive signals. Such statistical independence also holds for distortion signal terms at different chains, or at different time instance, i.e., they follow a spatially and temporally white statistics, see \cite[Subsection~II.C]{6177689}, and \cite[Subsection~II.D]{6177689}. Consequently from (\ref{eq_model_distortion_stat_2}) and (\ref{eq_model_distortion_stat_1}), the covariance of the aggregate noise-plus-residual-interference signal on Bob is obtained as 
\begin{align}
\ma{\Sigma}_{b}^{(n)} & = \mathbb{E} \left\{ \left(   \ma{n}_{b}^{(n)} +  \ma{z}_{b}^{(n)}  \right) \left(   \ma{n}_{b}^{(n)} +  \ma{z}_{b}^{(n)}  \right)^H \right\} \nonumber \\ & =  N_{b}^{(n)} \ma{I}_{M_{br}} + \text{trace}\left(\ma{W}^{(n)}\right) \ma{D}_{bb}^{(n)} {\ma{D}_{bb}^{(n)}}^{H}  \nonumber \\
& + \ma{H}_{bb}^{(n)}  \left( \kappa^{(n)} \sum_{n \in \mathcal{N}}   \text{diag} \left(\ma{W}^{(n)} \right)  \right)  {\ma{H}_{bb}^{(n)}}^{H} \nonumber \\ 
& +  \beta^{(n)} \text{diag} \left(  \sum_{n \in \mathcal{N}}  \ma{H}_{bb}^{(n)} \ma{W}^{(n)} {\ma{H}_{bb}^{(n)}}^H \right),
 \label{eq_model_aggregate_interference_covariance}
\end{align}  
where $\kappa^{(n)}$ ($\beta^{(n)}$) represents the transmit (receive) distortion coefficient relating the collective power of the intended transmit (receive) signal to the distortion signal variance in the $n$-th subcarrier\footnote{The distortion coefficients associated with different subcarriers may be different if, e.g., the subcarrier spacing is not equal over all bands, or the power spectral density of the distortion signals are not completely flat.}\cite{taghizadeh2017linear}.

It is worth mentioning that the impacts of the discussed inaccuracies, i.e., ${e}_{\text{t},l},{e}_{\text{r},l},{\ma{E}}_{bb}^{(n)}$, become significant for a FD transceiver due to the strong SI channel. For instance, the transmit distortion signals pass through the strong SI channel $\ma{H}_{bb}^{(n)}$ and become comparable to the desired signal from Alice which is passing through a much weaker channel $\ma{H}_{ab}^{(n)}$. Nevertheless, such inaccuracies are ignorable in the other links which do not involve an SI path, i.e., $\kappa \ll 1, \beta \ll 1$ and $\left\| \ma{E}_{bb}^{(n)}\right\|_F \ll \left\| \ma{H}_{bb}^{(n)}\right\|_F$.

\begin{figure}[!t] 
    \begin{center}
        \includegraphics[angle=0,width=0.85\columnwidth]{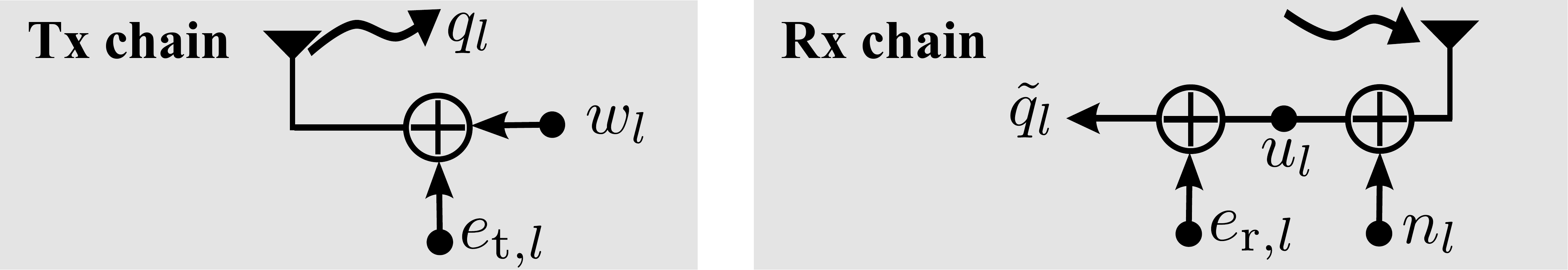}
    \caption{{{FD transceiver inaccuracy model. The impact of the limited transmitter/receiver dynamic range is modeled as additive distortion terms.}}} \label{fig:TransceiverAccuracyModel}
    \end{center} \vspace{-0mm} 
\end{figure}

\subsection{Transmit power constraints}
It is common to assume that the average transmit power from each device is practically limited due to, e.g., limited battery and energy storage. This is written as 
\begin{align} \label{model:PowerConstraints}
\text{tr} \left( \sum_{n\in \mathcal{N}} \ma{X}^{(n)}  \right) \leq X_{\text{max}}, \;\; \text{tr} \left( \sum_{n\in \mathcal{N}} \ma{W}^{(n)}  \right) \leq W_{\text{max}},
\end{align} 
where 
\begin{align} 
\ma{X}^{(n)}= \mathbb{E} \left\{ \ma{x}^{(n)} {\ma{x}^{(n)}}^H\right\} = \ma{V}^{(n)} {\ma{V}^{(n)}}^H 
\end{align}  
is the transmit covariance from Alice in the subcarrier $n$. Moreover, $X_{\text{max}},W_{\text{max}} \in \real^+$ respectively represent the maximum transmit power from Alice, and the maximum transmit power from Bob, i.e., jamming power. 

\subsection{Sum secrecy capacity}

The secrecy capacity of the defined system, for the subcarrier $n$ is written as
\begin{align} \label{model:I_sec}
\mathcal{I}_{\text{sec}}^{(n)} & = \left\{ \mathcal{I}_{ab}^{(n)} - \mathcal{I}_{ae}^{(n)}  \right\}^{+} \nonumber \\
& = \Bigg\{  \text{log}_2 \left| \ma{I}_d + {\ma{H}_{ab}^{(n)}}\ma{X}^{(n)}{\ma{H}_{ab}^{(n)}}^H  \left( \ma{\Sigma}_b^{(n)}\right)^{-1} \right| \nonumber  \\
& \;\;\;\;\;\;\;\; - \text{log}_2 \left| \ma{I}_d + {\ma{H}_{ae}^{(n)}} \ma{X}^{(n)}{\ma{H}_{ae}^{(n)}}^H  \left( \ma{\Sigma}_e^{(n)}\right)^{-1}\right| \Bigg\}^{+} ,
\end{align} 
where $\mathcal{I}_{\text{sec}}^{(n)} $ is the resulting secrecy rate in subcarrier $n$, and $\mathcal{I}_{ab}^{(n)}$, $\mathcal{I}_{ae}^{(n)}$ respectively represent the information capacity of Alice-Bob and Alice-Eve paths. In the above formulation, $\ma{\Sigma}_e^{(n)} = N_{e}^{(n)} \ma{I}_{M_{e}} + {\ma{H}_{be}^{(n)}}  \ma{W}^{(n)}  {\ma{H}_{be}^{(n)}}^H$ is the covariance of the received noise-plus-interference signal at Eve and $\ma{\Sigma}_b^{(n)}$ is calculated from (\ref{eq_model_aggregate_interference_covariance}). 

Consequently, the sum-secrecy capacity of the defined multi-carrier system is obtained as 
\begin{align} \label{model:I_sum}
\mathcal{I}_{\text{sum}} & = \sum_{n\in \mathcal{N}}   \mathcal{I}_{\text{sec}}^{(n)}.
\end{align}

\subsection{Remarks} 
\textit{i)}~Unlike the information-containing signal from Alice, the jamming signal from Bob contains artificial noise, see \cite[Equation~(5)]{7463025}. This is to prevent Eve to decode the jamming signal. \par
\textit{ii)}~We consider a case where Alice is not contributing in the jamming process. This is to simplify the task of Alice as a usual end-user device. The extension of the considered system to a setup with different jamming strategies is a goal of our future research. \par
\textit{iii)}~In this work we assume the availability of CSI on Alice-Bob, Alice-Eve, and Bob-Eve channels. Other than scenarios where the position of Eve is stationary and known, this assumption does not hold in practice. The sensitivity of the resulting system performance to the CSI accuracy is numerically evaluated in Section~\ref{sec:simulations}. 
%
%
%

\section{Sum Secrecy Rate Maximization} \label{sec:SSRM}
 In this part our goal is to maximize the defined sum secrecy rate of the system over all sub-carriers, see (\ref{model:I_sum}), considering the transmit power constraints for Alice and Bob, see (\ref{model:PowerConstraints}). The corresponding optimization problem is written as 
\begin{align}
\underset{\mathbb{X}, \mathbb{W}}{\text{max}} \;\;  &\;\;  \mathcal{I}_{\text{sum}} \;\;
{\text{s.t.}} \;\;    \text{(\ref{model:PowerConstraints})} \label{OptProblem1}
\end{align} 
where $\mathbb{X}$ ($\mathbb{W}$) represents the set of $\ma{X}^{(n)}\succeq 0$ ($\ma{W}^{(n)}\succeq 0$), $\forall n \in \mathcal{N}$. By recalling (\ref{model:I_sec}) and known matrix identities \cite[Eq.~(516)]{MCB:08} the defined problem is reformulated as
\begin{subequations}  \label{OptProblem2}
\begin{align}
\underset{\mathbb{X}, \mathbb{W}}{\text{max}} \;\;  &\;\;  \sum_{n \in \mathcal{N}} \Bigg\{  \text{log}_2 \left| \ma{\Sigma}_b^{(n)} + \ma{\Theta}_b^{(n)} \right| - \text{log}_2 \left| \ma{\Sigma}_b^{(n)} \right| \nonumber  \\
& \;\;\;\;\;\;\;\; - \text{log}_2 \left| \ma{\Sigma}_e^{(n)} + \ma{\Theta}_e^{(n)}   \right| + \text{log}_2 \left| \ma{\Sigma}_e^{(n)} \right| \Bigg\}^{+}  \\
{\text{s.t.}} \;\;    & \text{tr} \left( \ma{\Theta}  \right) \leq X_{\text{max}}, \;\; \text{tr} \left( \ma{\Sigma} \right) \leq W_{\text{max}},
\end{align} 
\end{subequations}
where $\ma{\Theta}:= \sum_{n\in \mathcal{N}} \ma{X}^{(n)}$, $\ma{\Theta}_b^{(n)} := {\ma{H}_{ab}^{(n)}}{\ma{X}^{(n)}} {\ma{H}_{ab}^{(n)}}^H$, and $\ma{\Theta}_e^{(n)} := {\ma{H}_{ae}^{(n)}}{\ma{X}^{(n)}}{\ma{H}_{ae}^{(n)}}^H$ are affine compositions of the Alice transmit covariance matrices $\ma{X}^{(n)}$. Moreover, $\ma{\Sigma} = \sum_{n\in \mathcal{N}} \ma{W}^{(n)}$, $\ma{\Sigma}_e^{(n)}$, and $\ma{\Sigma}_b^{(n)}$ are affine compositions of the transmit jamming covariance matrices $\ma{X}^{(n)}$. Nevertheless, the above problem is intractable, due to the $\{.\}^{+}$ operation in the defined secrecy value. Moreover, the maximization of difference of such $\text{log}()$ functions lead to a class of difference-of-convex (DC) problems which is jointly or separately a non-convex problem~\cite{Boyd:2004}. The following two lemmas transform the objective function into a more tractable form. 
\begin{lemma} \label{lemma_positiveOperation}
At the optimality of (\ref{OptProblem2}), the operator $\{.\}^{+}$ has no impact. 
\end{lemma}
\begin{proof}
The operator $\{.\}^{+}$ impacts the problem only when the value inside becomes negative, for at least one of the sub-carriers ${n \in \mathcal{N}}$. In such a situation, Alice and Bob can turn off transmission in the corresponding sub-carrier, i.e., choosing $\ma{X}^{(n)}=\ma{0}, \ma{W}^{(n)}=\ma{0}$, and contribute the reduced power to another sub-carrier with a positive $\mathcal{I}_{\text{sec}}^{(n)}$. Such action results in the enhancement of $\mathcal{I}_{\text{sum}}$ and contradicts the initial optimality assumption.   
\end{proof}

\begin{lemma} \label{lemma_logdetE}
Let $\ma{R}\in\compl^{l \times l}$ be a positive definite matrix. The maximization of the term $-\text{log} \left|\ma{R} \right|$ is equivalent in terms of the optimal $\ma{R}$ and objective value to
\begin{align} \label{lemma_logdetE_equivalance}
\underset{\ma{R}\succ 0, \ma{Q}\succ 0}{ \text{max}} \text{log}\left|\ma{Q} \right| - \text{tr}\left(  \ma{Q} \ma{R} \right)   + l,
\end{align} 
where $\ma{Q} \in \compl^{l \times l}$.
\end{lemma}
\begin{proof}
Since (\ref{lemma_logdetE_equivalance}) is an unconstrained concave maximization over $\ma{Q}$ for a fixed $\ma{R}$, the corresponding optimal $\ma{R}$ is obtained by equalizing the derivative of the objective function to zero. In this way we obtain $\ma{Q}^\star = \ma{R}^{-1}$. This equalizes the objective in (\ref{lemma_logdetE_equivalance}) to the term $-\text{log} \left|\ma{R} \right|$ at the optimality of $\ma{Q}$, which concludes the proof, see also \cite[Lemma~2]{5962875}. 
\end{proof}
Via the utilization of the defined lemmas, the problem (\ref{OptProblem2}) can be restructured as   
\begin{subequations}  \label{OptProblem3}
\begin{align}
\underset{\mathbb{X}, \mathbb{W}, \mathbb{Q}, \mathbb{L}}{\text{max}} \;\;  &\;\;  \sum_{n \in \mathcal{N}} \Bigg(  \text{log} \left| \ma{\Sigma}_b^{(n)} + \ma{\Theta}_b^{(n)} \right| + \text{log} \left| \ma{\Sigma}_e^{(n)} \right|    \label{OptProblem3_a}  \\
&  - \text{tr}\left( \ma{Q}^{(n)}   \ma{\Sigma}_b^{(n)}   \right) -    \text{tr}\left(  \ma{T}^{(n)}  \left( \ma{\Sigma}_e^{(n)} + \ma{\Theta}_e^{(n)}  \right) \right) \nonumber \\
&  + \text{log} \left| \ma{Q}^{(n)} \right| +  \text{log} \left| \ma{T}^{(n)} \right| \Bigg)  \\ 
{\text{s.t.}} \;\;\;\;    & \text{tr} \left( \ma{\Theta}  \right) \leq X_{\text{max}}, \;\; \text{tr} \left( \ma{\Sigma} \right) \leq W_{\text{max}},
\end{align} 
\end{subequations}
where $\text{log}()$ is natural logarithm, $\ma{Q}^{(n)} \in \compl^{M_{br}\times M_{br}}$ and $\ma{T}^{(n)} \in \compl^{M_{e}\times M_{e}}$ are introduced as auxiliary variables, and the sets $\mathbb{Q}$ and $\mathbb{T}$ are defined similar to that of $\mathbb{X}$, $\mathbb{W}$. Following Lemma~\ref{lemma_logdetE} optimal values of the auxilliary variables are obtained as 
\begin{align}
{\ma{T}^{(n)}}^\star &= \left( \ma{\Sigma}_e^{(n)} + \ma{\Theta}_e^{(n)}  \right)^{-1}, \label{solution_auxilliary_T} \\
{\ma{Q}^{(n)}}^\star &= \left( \ma{\Sigma}_b^{(n)} \right)^{-1}.  \label{solution_auxilliary_Q}
\end{align}

Please note that the obtained problem structure (\ref{OptProblem3}) is not a jointly convex problem. Nevertheless, it is a convex problem separately over the sets $\mathbb{X}, \mathbb{W}$ and $\mathbb{Q}, \mathbb{T}$. Hence, we follow an iterative coordinate ascend update where in each iteration the original problem is solved over a subset of the original variable space, see \cite[Subsection~2.7]{bertsekas1999nonlinear}. Firstly, the problem is solved over the variable sets $\mathbb{X}, \mathbb{W}$, assuming the auxiliary variables $\ma{Q}^{(n)}$ and $\ma{T}^{(n)}$ are fixed. This results in a convex sub-problem, where the optimum point is efficiently obtained using the MAX-DET algorithm \cite{Vandenberghe:1998}. Secondly, the problem is solved over the auxiliary variable sets $\mathbb{Q}, \mathbb{T}$, where the optimum solution is obtained in closed-form from (\ref{solution_auxilliary_T}), (\ref{solution_auxilliary_Q}). This procedure is repeated until a stable point is obtained, see Algorithm~\ref{mainalgorithm}. Please note that due to the monotonic increase of the objective (\ref{OptProblem3_a}) in each optimization iteration, and the fact that the system secrecy capacity is bounded from above, the proposed algorithm leads to a necessary convergence. The average convergence behavior of Algorithm~\ref{mainalgorithm} is numerically studied in Section~\ref{sec:simulations}.
              %

\subsection{Initialization}
In this section we discuss two initialization methods for the optimization problem in (\ref{OptProblem3}). 
\subsubsection{Uniform covariance with equal power initialization}
This simple initialization method initialize the covariance matrix of transmit signal by uniform covariance matrix with equal power, i.e., $\ma{Q} \leftarrow \epsilon\ma{I}$. In our case $\ma{Q}$ represents any matrix of $\ma{X}^{(n)},\ma{W}^{(n)},\forall n \in \mathcal{N}$.
\subsubsection{Optimal spatial beam initialization}
This initialization method aims to obtain optimal spatial beam, where the transmit signal is orientated to the desired receiver and prevent signal leakage to the undesired directions. This is defined as the following maximization
\begin{align}
    \underset{\ma{Q}\in \mathcal{H}}{\text{max}} \;\;  &\;\;  \frac{\text{tr}(\ma{F}\ma{Q}{\ma{F}}^{H}) + \nu_f}{\text{tr}(\ma{G}\ma{Q}{\ma{G}}^{H}) + \nu_g}, \;\;
    {\text{s.t.}} \;\;  \text{tr}(\ma{Q}) = 1, \label{OptIni}
\end{align}
where $\ma{Q}$ represents the normalized covariance matrix of the transmit signal, $\ma{F}$ and $\ma{G}$ are the desired and undesired channels, $\nu_f, \nu_g$ are the noise variances at the desired and undesired receivers, respectively. An optimal solution to (\ref{OptIni}) is obtained as 
\begin{align}
    \text{vec}(\ma{Q}^{\star \frac{1}{2}}) = \mathcal{P}_{\text{max}} \left( \left(\ma{I}\otimes\ma{G}^H \ma{G}+\nu_g\ma{I}\right)^{-1}  \left(\ma{I}\otimes\ma{F}^H \ma{F}+\nu_f\ma{I}\right) \right),
\end{align}
where $\mathcal{P}_{\text{max}}(\cdot)$ calculates the dominant eigenvector. The transmit power is uniform equally allocated. Please note that the above approach is applied separately for the initialization of the covariance matrix of information signal $(\ma{X}^{(n)})$ and jamming signal $(\ma{W}^{(n)})$ in each subcarrier. Specifically, for the design of $\ma{X}^{(n)}$, we set $\ma{F} \leftarrow \ma{H}_{ab}^{(n)}$ and $\ma{G} \leftarrow \ma{H}_{ae}^{(n)}$. For the design of $\ma{W}^{(n)}$ we set $\ma{F} \leftarrow \ma{H}_{be}^{(n)}$. The choice of $\ma{G}$ for $\ma{W}^{(n)}$ is related to the impact of distortion terms on Bob, which reflects the effect of the residual self-interference. The distortion power at Bob in $n$-th subcarrier can be written as
\begin{align}
    &\text{tr}\left(\kappa^{(n)}\ma{H}^{(n)}_{bb}\text{diag}\left(\ma{W}^{(n)}\right){\ma{H}^{(n)}_{bb}}^H\right) \nonumber\\ &+ \text{tr}\left(\beta^{(n)}\text{diag}\left(\ma{H}^{(n)}_{bb}\ma{W}^{(n)}{\ma{H}^{(n)}_{bb}}^H\right)\right)\nonumber \\
    &+  \text{tr}\left(\text{tr}\left(\ma{W}^{(n)}\right)\ma{D}_{bb}^{(n)}{\ma{D}_{bb}^{(n)}}^H\right)\nonumber \\
    & = \text{tr}\left(\tilde{\ma{H}}_{bb}^{(n)}\ma{W}^{(n)}\right), 
\end{align}
where $\tilde{\ma{H}}_{bb}^{(n)} = \kappa^{(n)}\text{diag}\left({\ma{H}^{(n)}_{bb}}^H\ma{H}^{(n)}_{bb}\right)+\beta^{(n)}{\ma{H}^{(n)}_{bb}}^H\ma{H}^{(n)}_{bb}+\text{tr}\left(\ma{D}_{bb}^{(n)}{\ma{D}_{bb}^{(n)}}^H\right)\ma{I}_{M_{bt}}$, which consequently results in the choice of $\ma{G}\leftarrow{\left(\tilde{\ma{H}}_{bb}^{(n)}\right)}^{\frac{1}{2}}$.

Please note that the optimization problem does not necessarily converge to the global optimum point based on any of above two initialization methods. In Algorithm~\ref{mainalgorithm} we apply the uniform covariance with equal power initialization. The performance of two initialization methods and the optimality gap are numerically compared and analyzed in Subsection~\ref{algorithmAnalysis} by examining multiple random initialization.

\subsection{Analytical computational complexity}
To analyze the computational complexity of arithmetic operations we consider the floating-point operations (FLOP)s \cite{hunger2005floating}. One FLOP represents a complex multiplication or a complex summation. Then the arithmetic operations for the calculation of $\mathbb{Q}_l$ and $\mathbb{T}_l$ via (\ref{solution_auxilliary_T}), (\ref{solution_auxilliary_Q}) with the inverse terms via Cholesky decomposition in Algorithm \ref{mainalgorithm} are in total
\begin{align*}
    \mathcal{O}&\Big(\gamma N \Big(M^3_e+M^3_{br}+M_{bt}Me(2M_{bt}+M_e) \\
    &+M_aM_e(2M_a+M_e)+M_{bt}M_{br}(2M_{bt}+3M_{br})\Big)\Big)
\end{align*}
FLOPs\cite{hunger2005floating}, where $\gamma$ is the total required iterations number until convergence. 

More dominant computational complexity of Algorithm \ref{mainalgorithm} is incurred in the steps of the determinant maximization, see Algorithm \ref{mainalgorithm}, Step 7. A general form of a MAX-DET problem is defined as
\begin{align}
    \underset{\ma{z}}{\text{min}} \; \ma{p}^{T}\ma{z} + \text{log}\left| \ma{Y}(\ma{z})^{-1}\right|, \; \text{s.t.} \; \ma{Y}(\ma{z})\succ, \ma{F}(\ma{z}) \succeq 0,      
    \label{MAX_DET}
\end{align}
where $\ma{z} \in \mathbb{R}^n$, and $\ma{Y}(\ma{z}) \in \mathbb{R}^{n_Y \times n_Y} := \ma{Y}_0 + \sum^n_{i=1}z_i\ma{Y_i}$ and $\ma{F}(\ma{z})\in \mathbb{R}^{n_F \times n_F}:= \ma{F}_0 + \sum^n_{i=1}z_i\ma{F}_i$. An upper bound to the computational complexity of the above problem is given as 
\begin{equation}\label{computationalComplexity}
    \mathcal{O}\left(\gamma \sqrt{n}(n^2+n^2_Y)n^2_F\right),
\end{equation}
where $\gamma$ is the total number of the required iterations until convergence, see \cite[Section 10]{Vandenberghe:1998}. In our problem $n=N(M^2_a + M^2_{bt})$ representing the dimension of real valued scalar variable space, and $n_Y = N(M_e + M_{br})$ and $n_F = N(M_a + M_{bt}) + 2$, representing the dimension of the determinant operation and the constraints space, respectively.

Please note that the above analysis only shows how the bounds on computational complexity are related to different problem dimensions. In practice the actual computational load may vary depending on the structure simplifications and used numerical solvers.

\begin{algorithm}[H] 
\small{	\begin{algorithmic}[1] 
\State{$\ell \leftarrow  {0}  ;  \;\;\;\; \text{set iteration number to zero}$}
\State{$\mathbb{X}_0 \leftarrow  \epsilon \ma{I}_{M_a} ; \;\;\;\; \text{initialization:~equal power in different sub-carriers}$}
            and uniform spatial beam
\State{$\mathbb{W}_0 \leftarrow   \ma{0}_{M_{bt}} ; \;\;\;\; \text{initialization:~initialize with zero jamming power}$}
\State{$\mathbb{Q}_{0}, \mathbb{T}_{0} \leftarrow   \ma{0} ; \;\;\;\; \text{initialization with zero matrices}$}

\Repeat 
\State{$\ell \leftarrow  \ell + 1;$}
\State{$\mathbb{X}_{\ell}, \mathbb{W}_{\ell} \leftarrow \text{solve MAX-DET (\ref{OptProblem3}), with \cite{Vandenberghe:1998}}$}
\State{$\mathbb{Q}_{\ell}, \mathbb{T}_{\ell} \leftarrow \text{calculate~(\ref{solution_auxilliary_T})~and~(\ref{solution_auxilliary_Q})}$}

\Until{$\text{a stable point, or maximum number of $\ell$ reached}$}

\State{\Return$\left\{\mathbb{X}_{\ell},\mathbb{W}_{\ell},\mathbb{Q}_{\ell},\mathbb{T}_{\ell}\right\}$}

  \end{algorithmic} } 
 \caption{\small{Iterative coordinate ascend method for sum secrecy rate maximization} } \label{mainalgorithm}
\end{algorithm}

\subsection{Optimal power allocation on Alice ($M_a = 1$)} 	
In this part we study a special case where Alice is equipped with a single antenna, and hence the problem of finding $\ma{X}^{(n)}$ reduces into a transmit power allocation problem among different sub-carriers. In this regard, we focus on finding an optimal transmit strategy from Alice assuming a known jamming strategy. This approach is particularly interesting where a joint design for Bob and Alice is not possible due to, e.g., feedback delay and overhead, computation complexity. Moreover, the obtained power allocation solution may also serve as a basis for a low-complexity sub-optimal design for a general case where Alice is facilitated with multiple antennas. 

The resulting power allocation problem on Alice, assuming a known jamming covariance is formulated as 
\begin{align} \label{OptProblem4}
\underset{{X}^{(n)}\geq0, \; \forall n\in\mathcal{N} }{\text{max}} \;\;  \;\;  \sum_{n\in\mathcal{N}} f_n\left({X}^{(n)}\right)   \;\;  {\text{s.t.}} \;\;   \sum_{n\in\mathcal{N}} X^{(n)} \leq X_{\text{max}}, 
\end{align} 
where
\begin{align} \label{f_n_definition}
f_n\left(X^{(n)}\right) := \text{log} \left( \frac{1+ \alpha^{(n)}X^{(n)}}{1+ \beta^{(n)}X^{(n)}}  \right).
\end{align} 
In the above formulation $f_n\left(X^{(n)}\right)$ is the realized secrecy capacity in the sub-carrier $n$, $\alpha^{(n)}:= {\ma{H}_{ab}^{(n)}}^H \left(\ma{\Sigma}_b^{(n)}\right)^{-1}{\ma{H}_{ab}^{(n)}}$ and $\beta^{(n)}:= {\ma{H}_{ae}^{(n)}}^H \left(\ma{\Sigma}_e^{(n)}\right)^{-1}{\ma{H}_{ae}^{(n)}}$, where $\alpha^{(n)}, \beta^{(n)} \in \real^{+}$. Note that a similar power allocation approach for sum secrecy rate maximization, in the context of HD broadcast multi-carrier systems is studied in \cite{4652697, 5872025}. Nevertheless, due to the presence of FD jamming in our system, the impact of the residual SI at Bob as well as the impact of the received jamming signal at Eve are respectively incorporated in $\alpha^{(n)}$ and $\beta^{(n)}$. To obtain an optimal solution of (\ref{OptProblem4}) we consider the Lagrangian function of the objective function in (\ref{OptProblem4}):
\begin{align}
    \mathcal{L}\left(\mathbb{X},\lambda,\boldsymbol{\tau}\right) = &\sum_{n\in \mathcal{N}}f_n(X^{(n)}) + \lambda\left(X_\text{max} - \sum_{n\in \mathcal{N}} X^{(n)} \right) \nonumber \\
    &+ \sum_{n\in \mathcal{N}} \tau^{(n)}X^{(n)},
\end{align}
where $\lambda$ and $\boldsymbol{\tau}$, which is the set of $\tau^{(n)}, n\in \mathcal{N}$, are the Lagrange multipliers for the inequality constrains.

It is observable from (\ref{f_n_definition}) that for $\alpha^{(n)} \leq \beta^{(n)}$ we have ${X^{(n)}}^\star = 0$. On the other hand, for $\alpha^{(n)} > \beta^{(n)}$, the function $f_n\left(X^{(n)}\right)$ is a concave and increasing composition of a concave and increasing function in $X^{(n)}$, and hence is a concave function, see \cite[Subsection~3.2.4]{Boyd:2004}. As a result, we obtain the necessary and sufficient optimality conditions of (\ref{OptProblem4}) by via the corresponding Karush-Kuhn-Tucker (KKT) conditions:
\begin{subequations}
\begin{align}
    \frac{\partial \mathcal{L}\left(\mathbb{X},\lambda,\boldsymbol{\tau}\right)}{\partial X^{(n)}} &= 0 \;, \quad \forall n \in\mathcal{N},\label{kkt_a} \\
    X^{(n)} &\geq 0 \;, \quad \forall n \in\mathcal{N}, \\
    X_{\text{max}} - \sum_{n\in \mathcal{N}}X^{(n)} &\geq 0 \;, \\
    \lambda &\geq 0 \;, \\
    \tau^{(n)} &\geq 0 \;, \quad \forall n \in\mathcal{N}, \\
    \lambda\left(X_\text{max} - \sum_{n \in \mathcal{N}}X^{(n)}\right) &= 0 \;, \\
    \tau^{(n)}X^{(n)} &= 0 \;, \quad \forall n \in \mathcal{N} \label{kkt_g}. 
\end{align}
\end{subequations}
The following lemma reveals an important property of the allocated power values at the optimality.   
\begin{lemma}
Let $\mathcal{N}_0$ be the set of sub-carriers with zero allocated power at the optimality, i.e., ${X^{(n)}}^\star = 0, \forall n\in \mathcal{N}_0$. Then we have 
\begin{align} \label{lemma_3_eq}
\frac{\partial f_{n} \left({X^{(n)}}^\star \right)}{\partial X^{(n)}} = \lambda, \;\; \forall n \in \mathcal{N}\setminus \mathcal{N}_0.
\end{align}
\end{lemma}
\begin{proof}
For the sub-carrier with ${X^{(n)}}^\star > 0$, we have $\tau^{(n)}=0$, due to (\ref{kkt_g}). Moreover, from (\ref{kkt_a}) we calculate $\frac{\partial f_{n} \left({X^{(n)}}^\star \right)}{\partial X^{(n)}} - \lambda + \tau^{(n)} = 0$. The two aforementioned arguments conclude the proof.
\end{proof}
The above lemma follows the interesting intuition that for the sub-carriers with positive allocated power, the slope of the objective function should be equal. This is expected, since if a slope of the objective is not equal for different sub-carriers, we can take power from the sub-carrier with smaller slope and reallocate it to the sub-carrier with a higher slope in the objective. 

From (\ref{lemma_3_eq}) we obtain a water-filling solution structure which can be expressed as in (\ref{Waterfilling_solution}), where $\lambda>0$ holds the concept of water level, c.f.~\cite[Equation~(17)]{5872025}. This identity shows that at the optimum values of ${X^{(n)}}^\star$ can be uniquely calculated for all sub-carriers, once the value of $\lambda^{\star}$ is obtained. Moreover we have 
\begin{align} \label{lambda_limits}
    0 \leq \lambda^{\star} \leq \left( \underset{n}{\text{max}} \frac{\alpha^{(n)} - \beta^{(n)} }{\left( 1+ \alpha^{(n)} X_{\text{max}} \right) \left( 1+ \beta^{(n)} X_{\text{max}} \right)} \right) =: \lambda_{\text{max}},
\end{align} 
which provides a feasible range for $\lambda^{\star}$. Hence, by choosing $\lambda$ as a search variable and performing a bi-section search, the optimal power allocation solution can be obtained with a water-filling procedure, see Algorithm~\ref{waterFilling} for the detailed procedure.

\begin{figure*}[!ht]
\normalsize
\begin{align} \label{Waterfilling_solution}
{X^{(n)}}^\star = \frac{1}{2}\left\{ -\left( \frac{1}{\beta^{(n)}} + \frac{1}{\alpha^{(n)}} \right) + \sqrt{ \left( \frac{1}{\beta^{(n)}} + \frac{1}{\alpha^{(n)}} \right)^2 -4 \left( \frac{1}{\alpha^{(n)} \beta^{(n)}}  -\frac{1}{\lambda^\star}\left( \frac{1}{\beta^{(n)}} - \frac{1}{\alpha^{(n)}} \right)\right) }\right\}^{+}
\end{align}   
\hrulefill
\vspace*{-0mm}
\end{figure*}

\begin{algorithm}
\caption{Water-filling optimization algorithm based on binary search}\label{waterFilling}
\begin{algorithmic}[1]
    \State $h \gets \lambda_{\text{max}}$, see (\ref{lambda_limits})
    \State $l \gets 0$, see (\ref{lambda_limits})
    \State \textbf{repeat}
        \State \qquad $\lambda \gets (h+l) / 2$
        \State \qquad $X^{(n)} \gets $ see (\ref{Waterfilling_solution})
        \State \qquad $\tilde{X} \gets \sum_{n \in \mathcal{N}}X^{(n)}$
        \State \qquad \textbf{if} $(X_{\text{max}} - \tilde{X}) < 0$ \textbf{then}
            \State \qquad \qquad $l \gets \lambda$
        \State \qquad \textbf{else}
            \State \qquad \qquad $h \gets \lambda$
        \State \qquad \textbf{end if}
    \State \textbf{until} $0 \leq X_{\text{max}} - \tilde{X} < \epsilon_0$ 
\State \textbf{return} $X^{(n)}$
\end{algorithmic}
\end{algorithm}

\section{Secure Bidirectional Full-Duplex Communication} 
\label{sec:extendedsolution}
\begin{figure}[!t] 
    \begin{center}
        \includegraphics[angle=0,width=0.90\columnwidth]{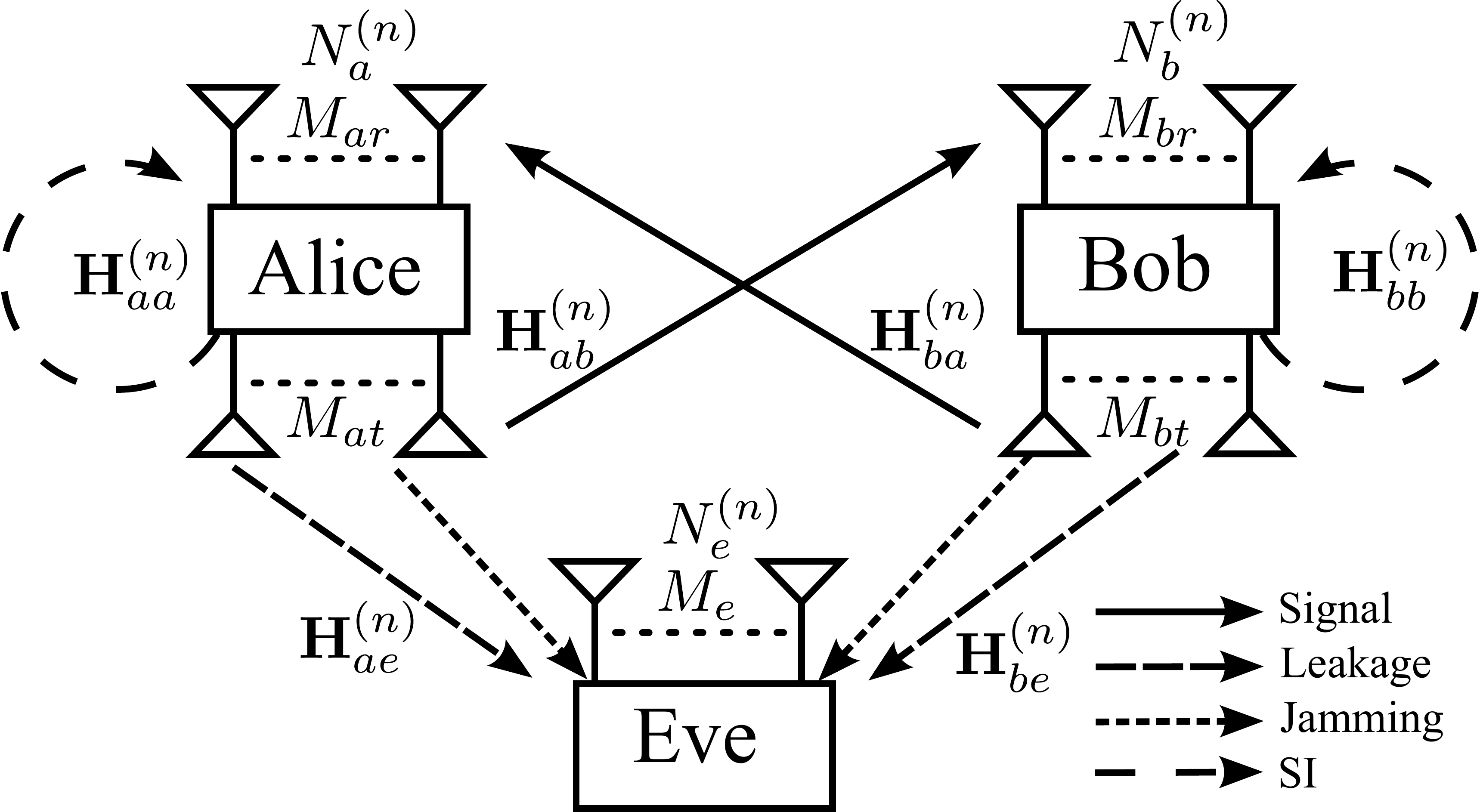}
    \caption{{{The studied bi-directional multi-carrier wiretap channel, where both Alice and Bob are FD nodes and able to jamming Eve.}}} \label{fig:bothFDmodel}
    \end{center} \vspace{-0mm} 
\end{figure}

The proposed solution in Section~\ref{sec:SSRM} acts as an efficient optimization framework for the underlying multi-carrier system defined in Section~\ref{sec:model}, following iterative updates using closed-form expressions (\ref{solution_auxilliary_T}) and (\ref{solution_auxilliary_Q}), as well as the MAX-DET algorithm. In this part we extend the same framework to support a more general setup where both Alice and Bob are capable of FD operation. This includes a bi-directional data communication between Alice and Bob, as well as the jamming capability at both nodes. The bi-directional system can improve the performance since it leads to a higher spectral efficiency \cite{6177689}, and the jamming signal from each node (Alice or Bob) can degrade Eve's reception quality at both directions. 

In order to update our notation, we denote the number of transmit (receive) antennas at Alice, Bob to Alice and Alice's SI channel in $n$-th sub-carrier as $M_{at}, M_{ar}, \ma{H}_{ba}^{(n)} $ and $ \ma{H}_{aa}^{(n)}$, respectively, see Fig. \ref{fig:bothFDmodel}. Moreover, we denote $\ma{s}^{(n)}_X, \ma{x}^{(n)}_X, \ma{w}^{(n)}_X, \ma{V}^{(n)}_X, \ma{X}^{(n)}_X, \ma{W}^{(n)}_X$ as the same signal types as in Section~\ref{sec:model}, but specific for the node $X$, such that $X\in\{a,b\}$. Thus the transmit signal from each node is updated as $\ma{u}^{(n)}_X = \ma{x}^{(n)}_X + \ma{w}^{(n)}_X$ with $\ma{x}^{(n)}_X = \ma{V}^{(n)}_X \ma{s}^{(n)}_X$, containing both the information and jamming signal from each node. Then the received interference-plus-noise covariance matrix in $n$-th subcarrier at Alice, Bob and Eve are updated as
\begin{align}
    \tilde{\ma{\Sigma}}_{\mathcal{X}}^{(n)} & =  N_{\mathcal{X}}^{(n)} \ma{I}_{M_{\mathcal{X}r}} +  \ma{H}_{\alpha_\mathcal{X}}^{(n)}\ma{W}_\mathcal{X}^{(n)}{\ma{H}_{\alpha_\mathcal{X}}^{(n)}}^H \nonumber \\ & +\text{trace}\left(\ma{X}_\mathcal{X}^{(n)}+\ma{W}_\mathcal{X}^{(n)}\right) \ma{D}_{\mathcal{X}\mathcal{X}}^{(n)} {\ma{D}_{\mathcal{X}\mathcal{X}}^{(n)}}^{H}  \nonumber \\
    & + \ma{H}_{\mathcal{X}\mathcal{X}}^{(n)}  \left( \kappa_\mathcal{X}^{(n)} \sum_{n \in \mathcal{N}}   \text{diag} \left(\ma{X}_\mathcal{X}^{(n)}+\ma{W}_\mathcal{X}^{(n)} \right)  \right)  {\ma{H}_{\mathcal{X}\mathcal{X}}^{(n)}}^{H} \nonumber \\ 
    & +  \beta_\mathcal{X}^{(n)} \text{diag} \left(  \sum_{n \in \mathcal{N}}  \ma{H}_{\mathcal{X}\mathcal{X}}^{(n)} \left(\ma{X}_\mathcal{X}^{(n)}+\ma{W}_\mathcal{X}^{(n)}\right) {\ma{H}_{\mathcal{X}\mathcal{X}}^{(n)}}^H \right), \label{eq_model_aggregate_interference_covariance_BI_b}\\
    \tilde{\ma{\Sigma}}_{e}^{(n)} & = N_e^{(n)} \ma{I}_{M_{e}} + \ma{H}_{ae}^{(n)}\ma{W}_a^{(n)}{\ma{H}_{ae}^{(n)}}^H +  \ma{H}_{be}^{(n)}\ma{W}_b^{(n)}{\ma{H}_{be}^{(n)}}^H, \label{eq_model_aggregate_interference_covariance_BI_e}
\end{align}
where $\mathcal{X} \in \{a,b\}$, $\alpha_a = ab$, $\alpha_b = ba$. Moreover $\kappa_{a}^{(n)}(\kappa_b^{(n)}) \in \mathbb{R}^{+}$ and $\beta_a^{(n)}(\beta_b^{(n)}) \in \mathbb{R}^{+}$ are the transmit and receive distortion coefficient at Alice(Bob) in the $n$-th subcarrier, $N_a^{(n)}$ represents the thermal noise variance at Alice in the $n$-th subcarrier.  Please note that in (\ref{eq_model_aggregate_interference_covariance_BI_e}) we consider a worst case scenario where the information signal from Alice and Bob as interference can be decoded by Eve\cite{6781609}. The defined system secrecy rate is hence written as 
\begin{align} 
    \tilde{\mathcal{I}}_{\text{sec}}^{(n)}  = \left\{ \tilde{\mathcal{I}}_{ab}^{(n)} - \tilde{\mathcal{I}}_{ae}^{(n)}  \right\}^{+} + \left\{ \tilde{\mathcal{I}}_{ba}^{(n)} - \tilde{\mathcal{I}}_{be}^{(n)}  \right\}^{+},
\end{align} 
where $\tilde{\mathcal{I}}_{ab}^{(n)}-\tilde{\mathcal{I}}_{ae}^{(n)}$ is obtained by applying (\ref{eq_model_aggregate_interference_covariance_BI_b}) and (\ref{eq_model_aggregate_interference_covariance_BI_e}) into (\ref{model:I_sec}), and
\begin{align}
    \tilde{\mathcal{I}}_{ba}^{(n)}-\tilde{\mathcal{I}}_{be}^{(n)} = \;&\text{log}_2 \left| \ma{I}_d +  {\ma{H}_{ba}^{(n)}}\ma{X}_{b}^{(n)}{\ma{H}_{ba}^{(n)}}^H  \left( \tilde{\ma{\Sigma}}_a^{(n)}\right)^{-1}\right| \nonumber \\  -\; &\text{log}_2 \left| \ma{I}_d + {\ma{H}_{be}^{(n)}} \ma{X}_{b}^{(n)} {\ma{H}_{be}^{(n)}}^H  \left( \tilde{\ma{\Sigma}}_e^{(n)}\right)^{-1} \right| ,
\end{align}
and the sum secrecy rate is defined the same as before.

\subsection{Bidirectional sum secrecy rate maximization}
Similarly to Section~\ref{sec:SSRM}, to maximize the sum secrecy rate in bidirectional communication system the optimization problem is written as 
\begin{subequations}  \label{BiOpt}
\begin{align}
    \underset{\tilde{\mathbb{X}}, \tilde{\mathbb{W}}}{\text{max}} \;\;  &\;\;  \sum_{n \in \mathcal{N}} \tilde{\mathcal{I}}_{\text{sec}}^{(n)}  \\
    {\text{s.t.}} \;\;    & \quad \text{tr} \left( \tilde{\ma{\Theta}}_a  \right) \leq P_{A,\text{max}}, \\ & \quad \text{tr} \left( \tilde{\ma{\Theta}}_b \right) \leq P_{B,\text{max}},
\end{align}
\end{subequations}
where $\tilde{\mathbb{X}}$ ($\tilde{\mathbb{W}}$) represents the set of $\ma{X}_X^{(n)} \succeq 0$ ($ \ma{W}_X^{(n)} \succeq 0$), $\forall n \in \mathcal{N}$, $\tilde{\ma{\Theta}}_X := \sum_{n\in \mathcal{N}}\left( \ma{X}_X^{(n)} + \ma{W}_X^{(n)}\right)$, $X \in \{a,b\}$ and $P_{A,\text{max}}, P_{B,\text{max}} \in \real^{+}$ represent the maximum transmit power of Alice and Bob.  

It is observed that the optimization problem in (\ref{BiOpt}) holds a similar mathematical structure in relation to the transmit covariance matrices, i.e., $\ma{X}_X^{(n)},\ma{W}_X^{(n)}$, $X \in \{a,b\}$, $\forall n \in \mathcal{N}$ as addressed for (\ref{OptProblem2}). Hence a similar procedure as in the Algorithm~\ref{mainalgorithm} following the result of the Lemma \ref{lemma_positiveOperation} and \ref{lemma_logdetE} is employed to obtain an optimal solution. The computational complexity of the steps of the determinant maximization is obtained similar to (\ref{computationalComplexity}), where $n=2N(M^2_{at} + M^2_{bt})$, $n_Y = N(M_e + M_{ar} + M_{br})$ and $n_F = 2N(M_{at} + M_{bt}) + 2$.

\section {Simulation Results}\label{sec:simulations}

\begin{figure}[!t] 
    \begin{center}
        \includegraphics[angle=0,width=0.89\columnwidth]{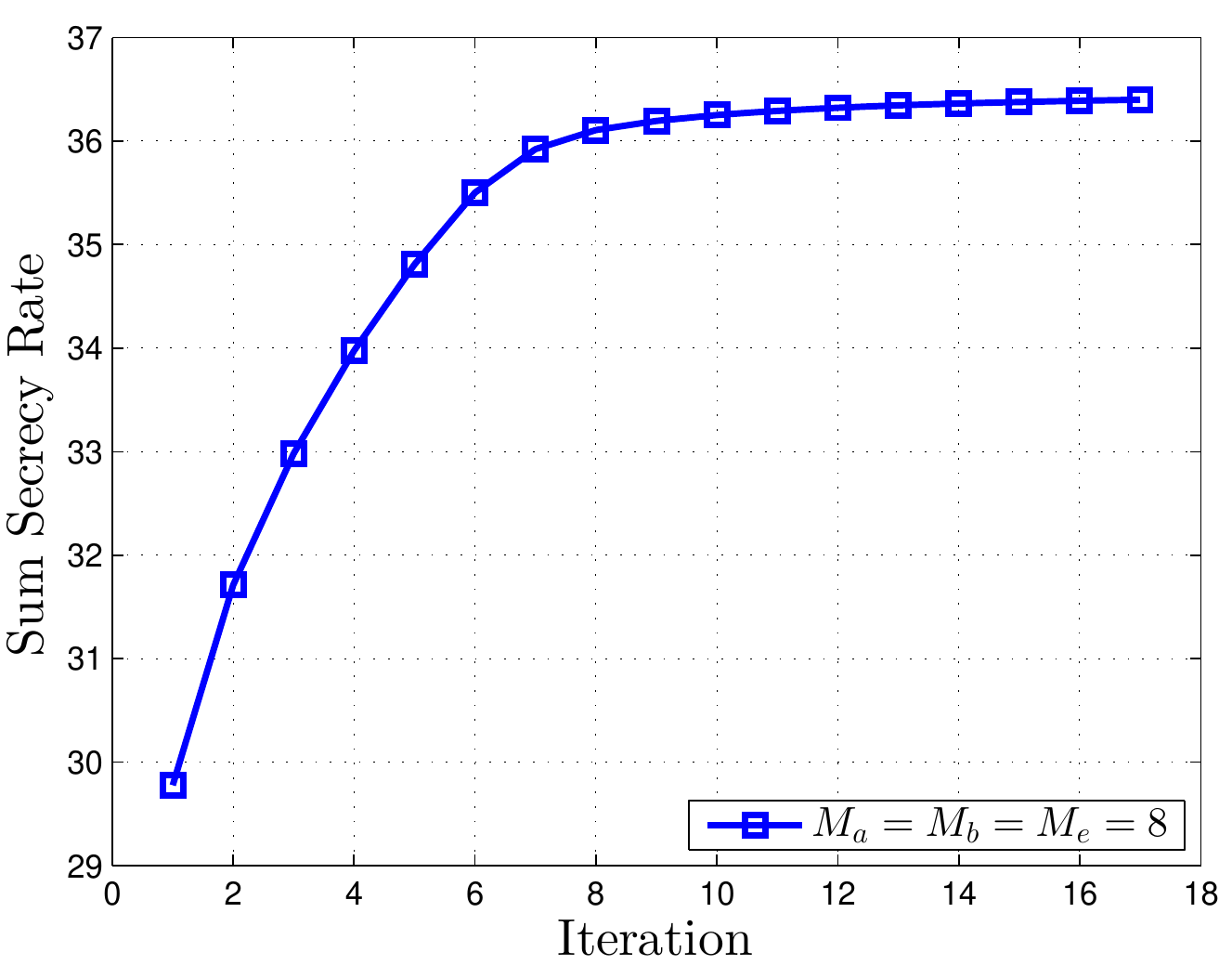}
    \caption{{{Average convergence behavior of the proposed iterative method. The proposed method converges to a stable point within 10-20 optimization iterations. }}} \label{fig:convergence}
    \end{center}
\end{figure}

\begin{figure}[!t] 
    \begin{center}
        \includegraphics[angle=0,width=0.89\columnwidth]{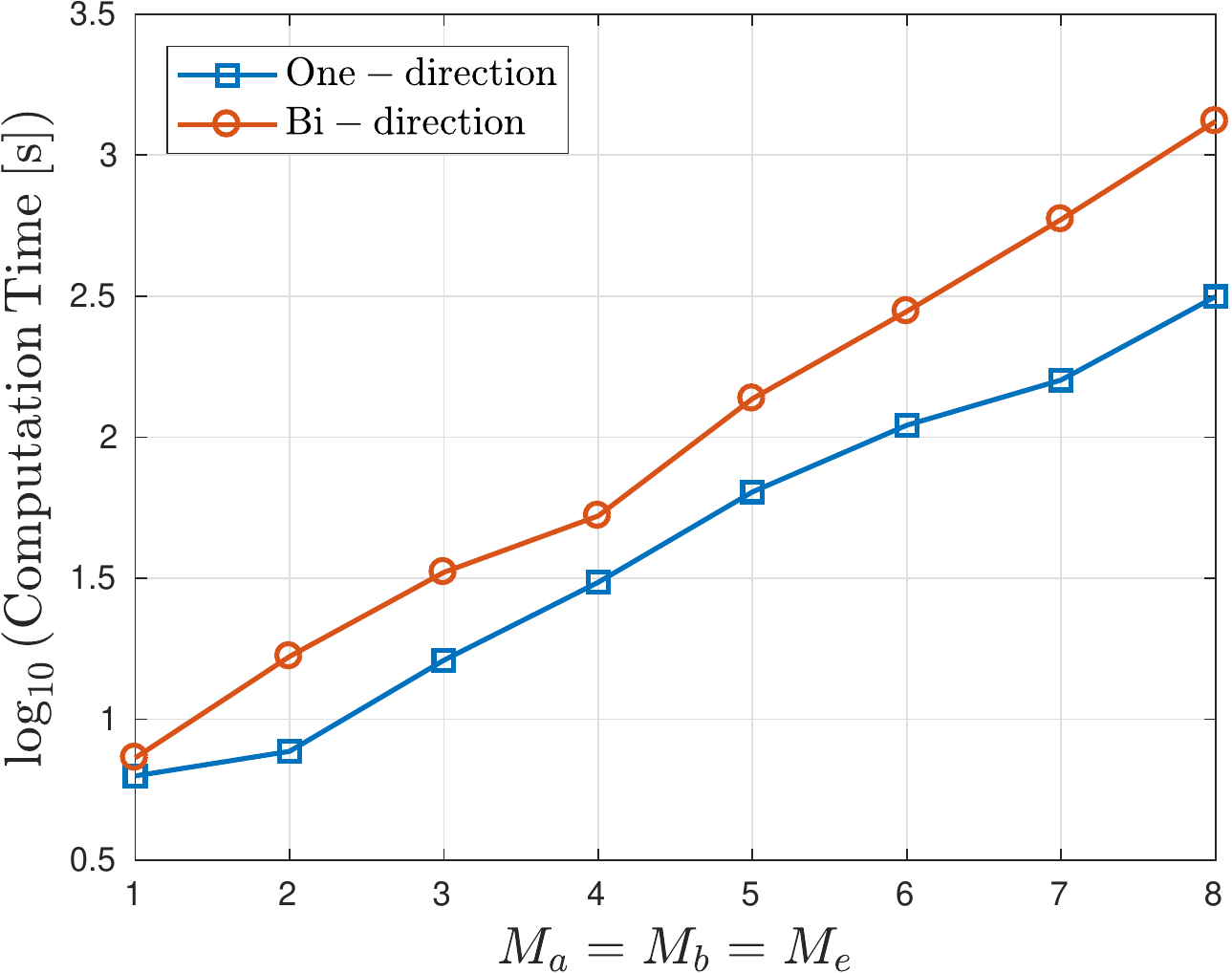}
    \caption{{{Average computation time of the proposed iterative method. }}} \label{fig:cpu}
    \end{center}
\end{figure}

\begin{figure}[!t] 
    \begin{center}
        \includegraphics[angle=0,width=0.89\columnwidth]{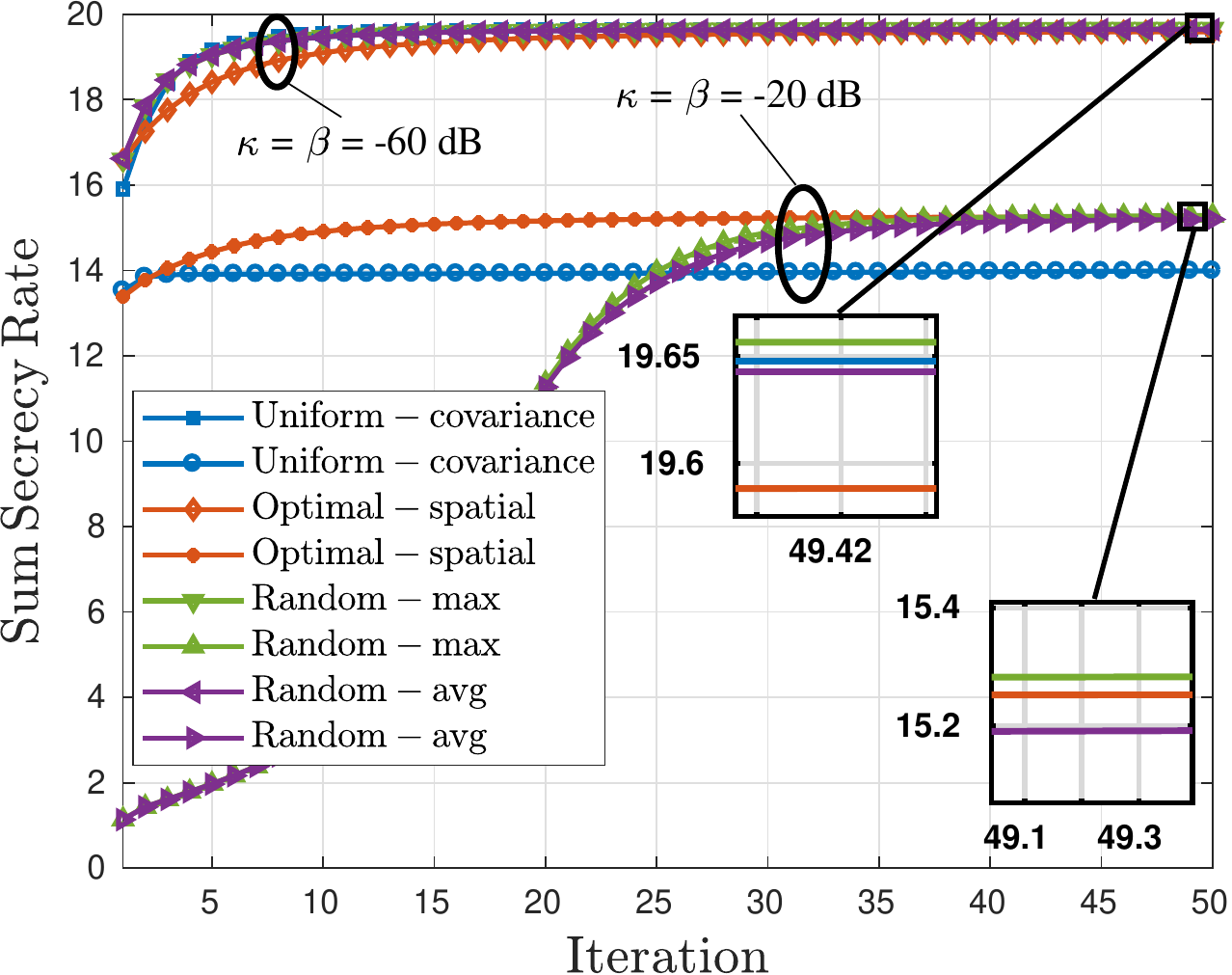}
    \caption{{{Impact of initialization on the proposed iterative method. }}} \label{fig:ini}
    \end{center}
\end{figure}

\begin{figure}[!t] 
    \begin{center}
        \includegraphics[angle=0,width=0.89\columnwidth]{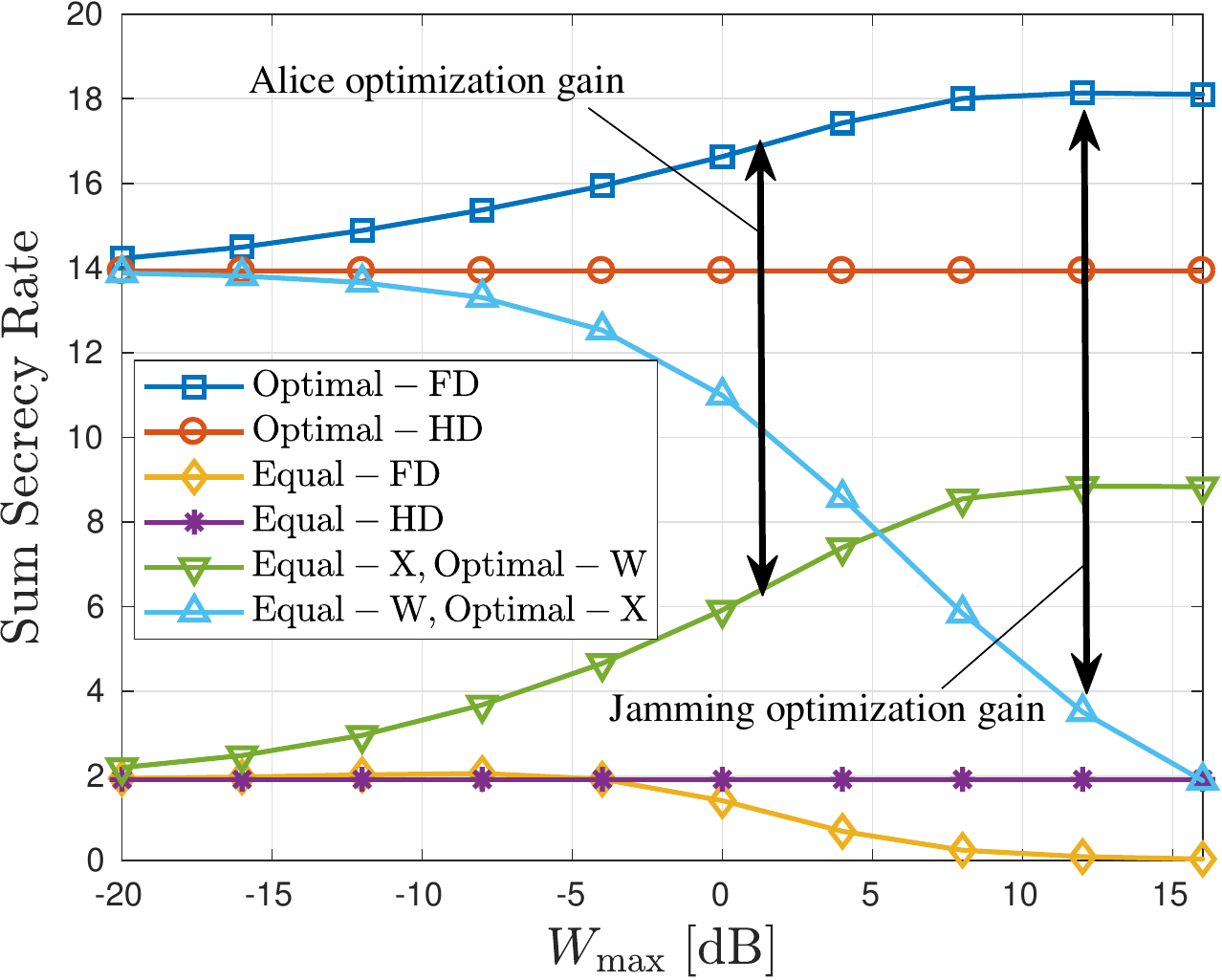}
    \caption{{{Obtained sum secrecy rate vs. maximum jamming power from Bob.}}} \label{fig:wmax}
    \end{center} 
\end{figure}

\begin{figure}[!t] 
    \begin{center}
        \includegraphics[angle=0,width=0.89\columnwidth]{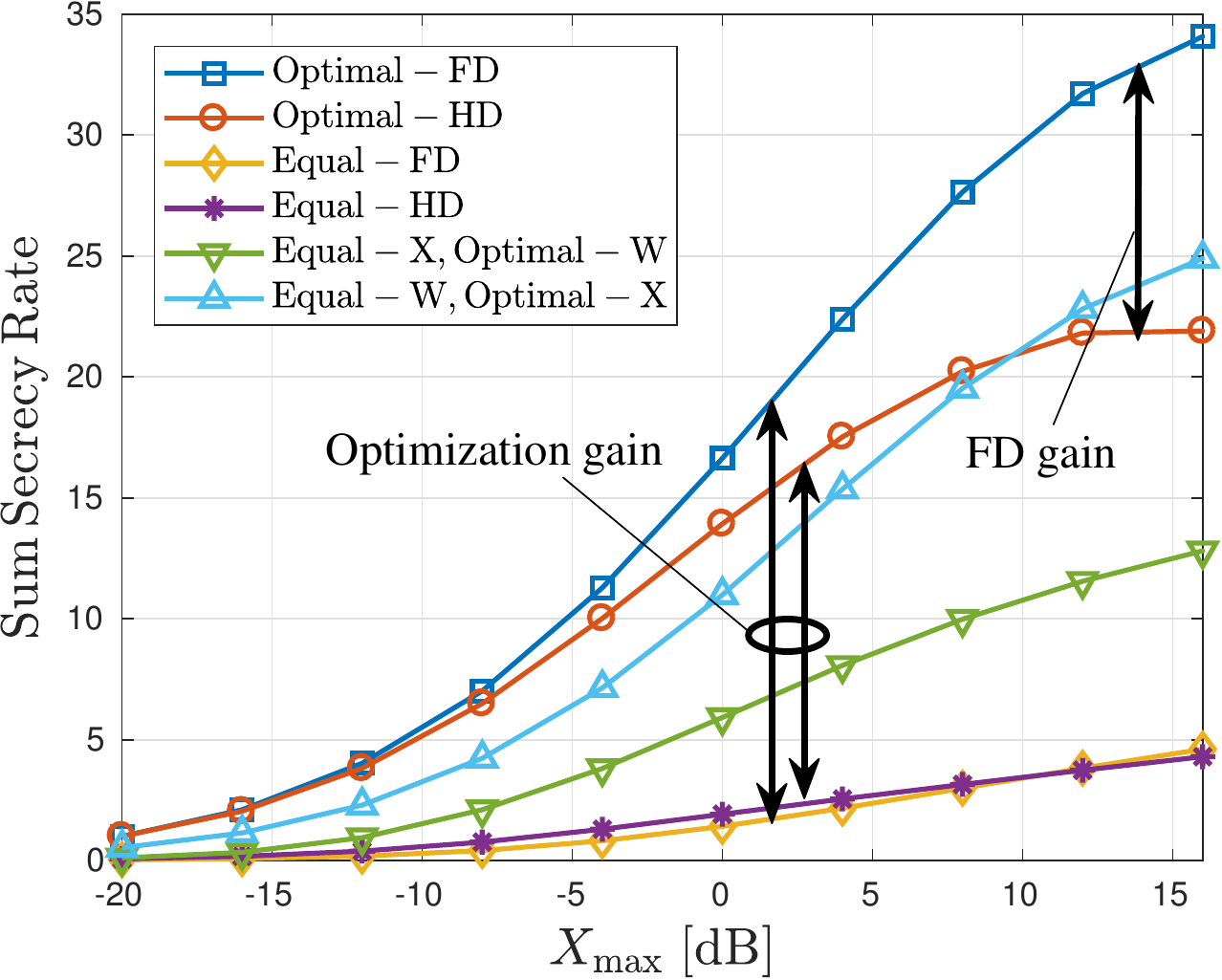}
    \caption{{{Obtained sum secrecy rate vs. maximum transmit power from Alice. }}} \label{fig:xmax}
    \end{center} 
\end{figure}  

\begin{figure}[!t] 
    \begin{center}
        \includegraphics[angle=0,width=0.89\columnwidth]{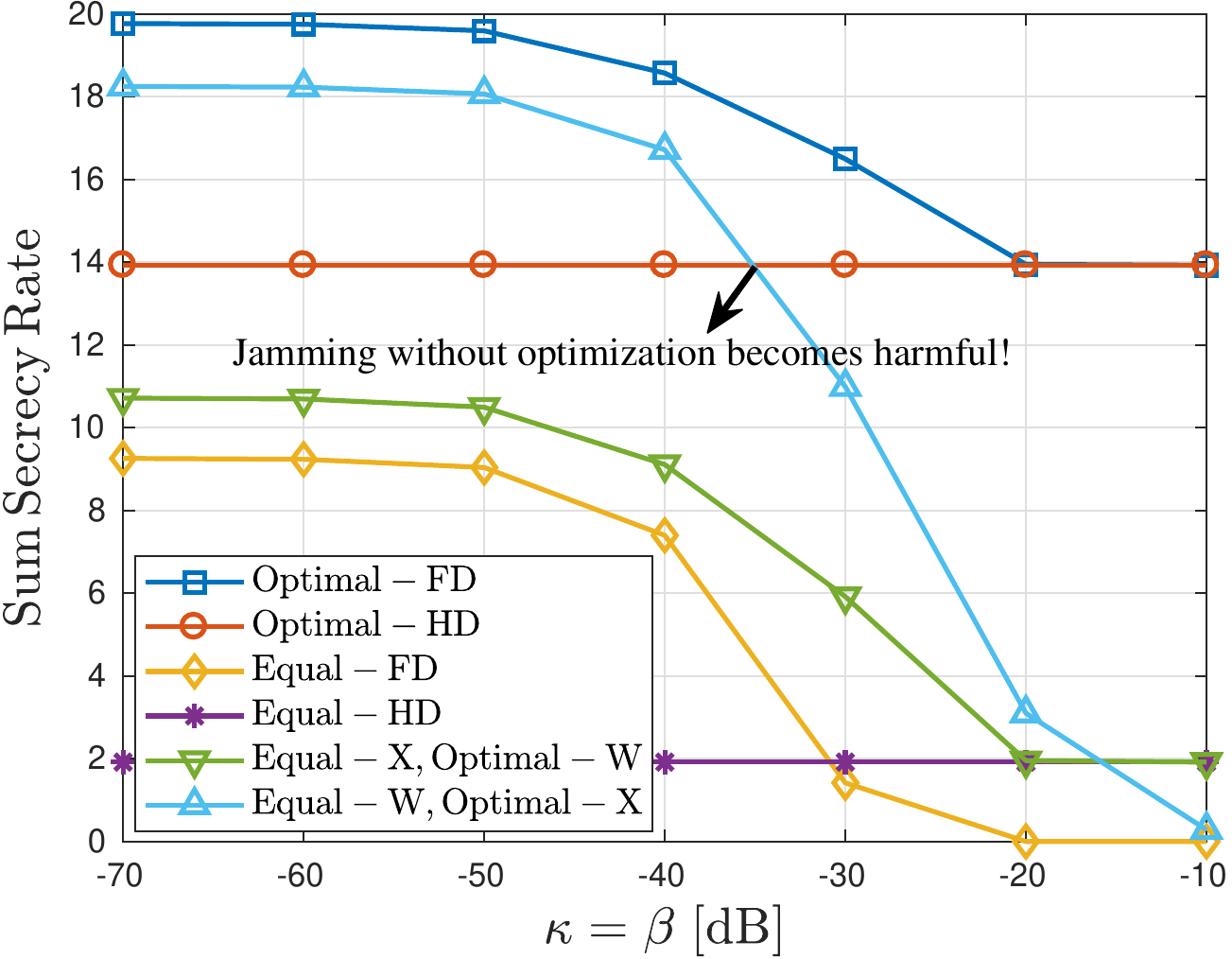}
    \caption{{{Obtained sum secrecy rate vs. transceiver dynamic range $\kappa = \beta$. }}} \label{fig:kappa}
    \end{center} 
\end{figure} 

\begin{figure}[!t] 
    \begin{center}
        \includegraphics[angle=0,width=0.89\columnwidth]{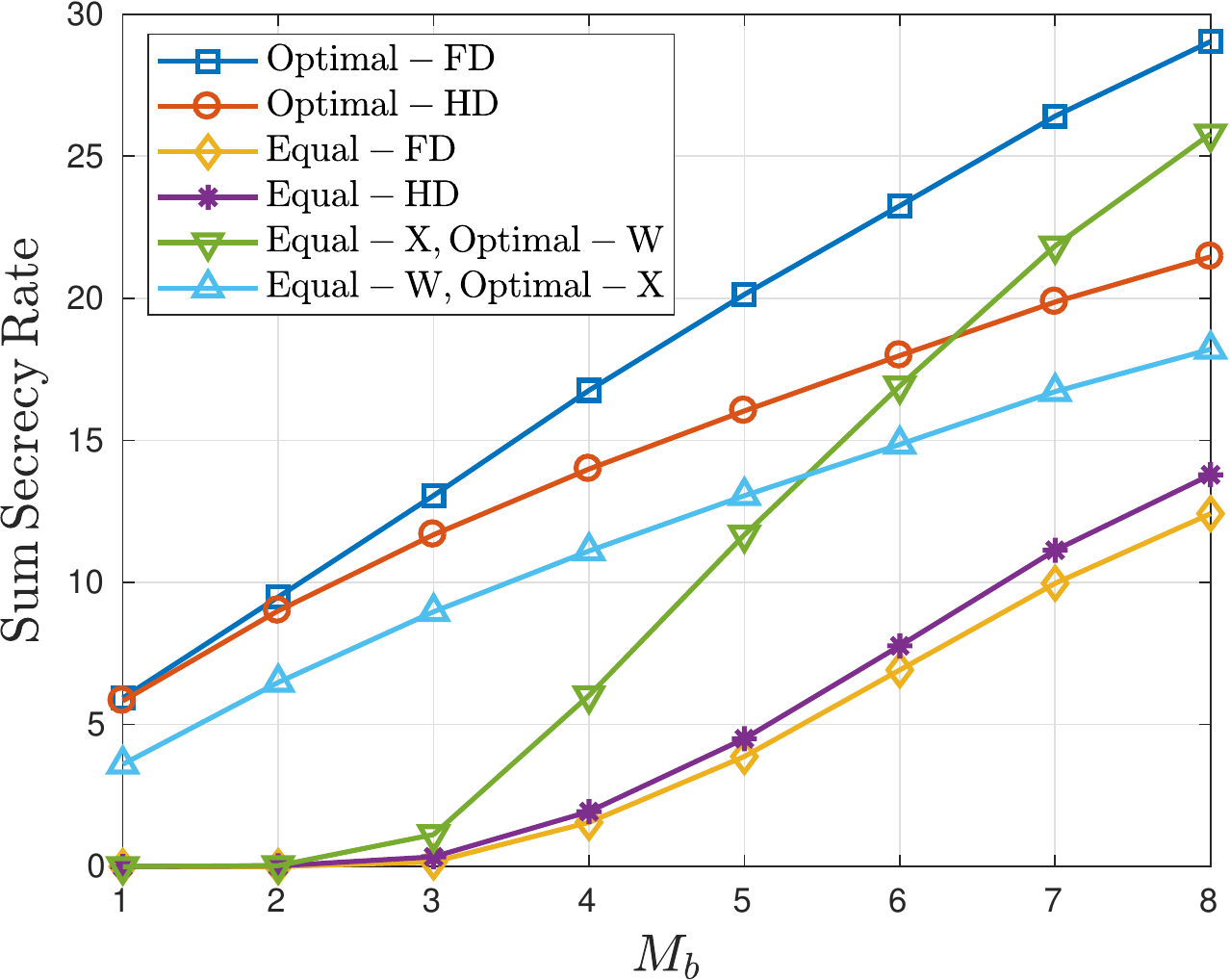}
    \caption{{{Obtained sum secrecy rate vs. number of the transmit/receive antennas at Bob $M_{b}=M_{bt}= M_{br}$. }}} \label{fig:Antenna_Bob2}
    \end{center} 
\end{figure} 

\begin{figure}[!t] 
    \begin{center}
        \includegraphics[angle=0,width=0.89\columnwidth]{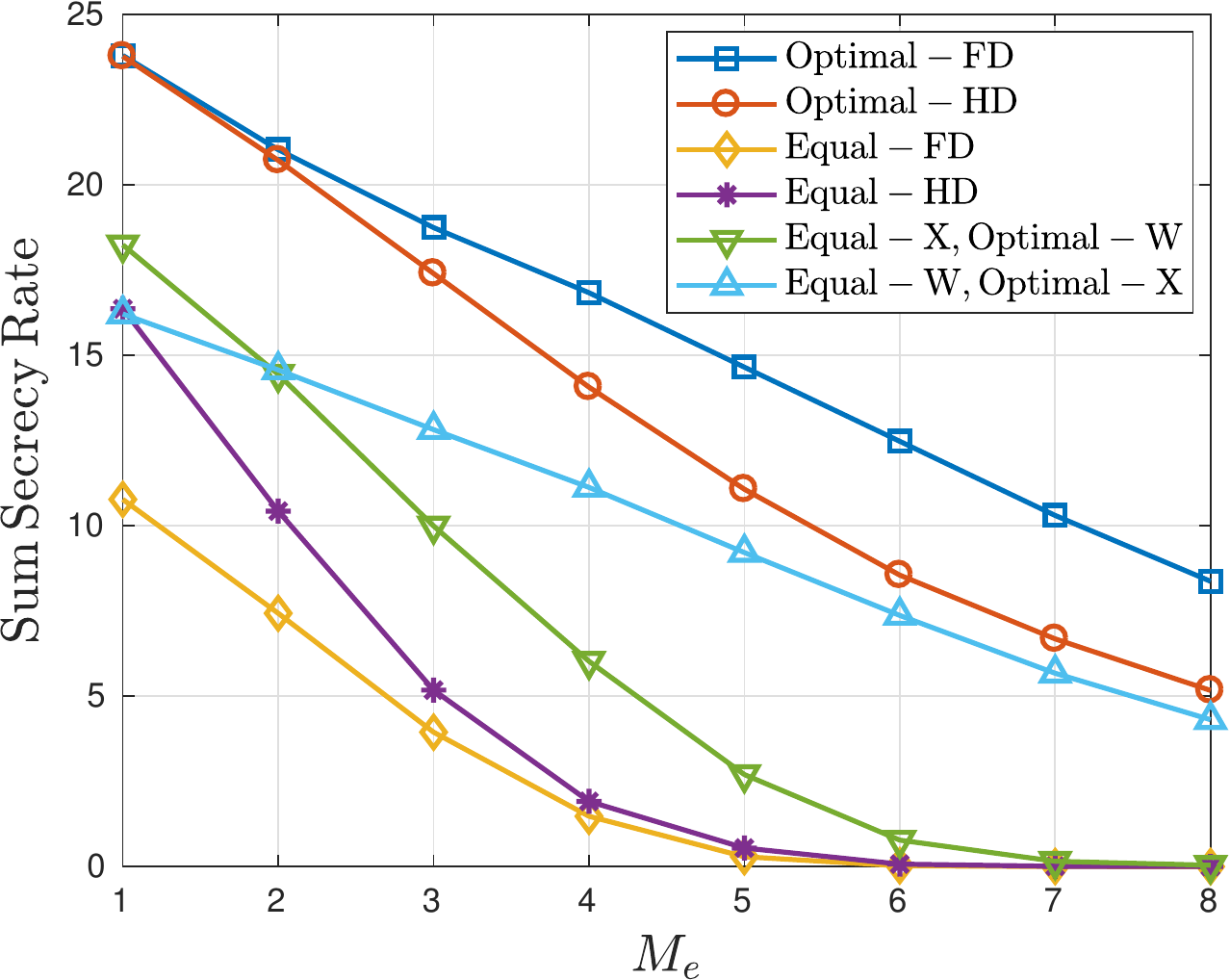}
    \caption{Obtained sum secrecy rate vs. number of the receive antennas at Eve.} \label{fig:Antenna_Eve2}
    \end{center} 
\end{figure} 

\begin{figure}[!t] 
    \begin{center}
        \includegraphics[angle=0,width=0.89\columnwidth]{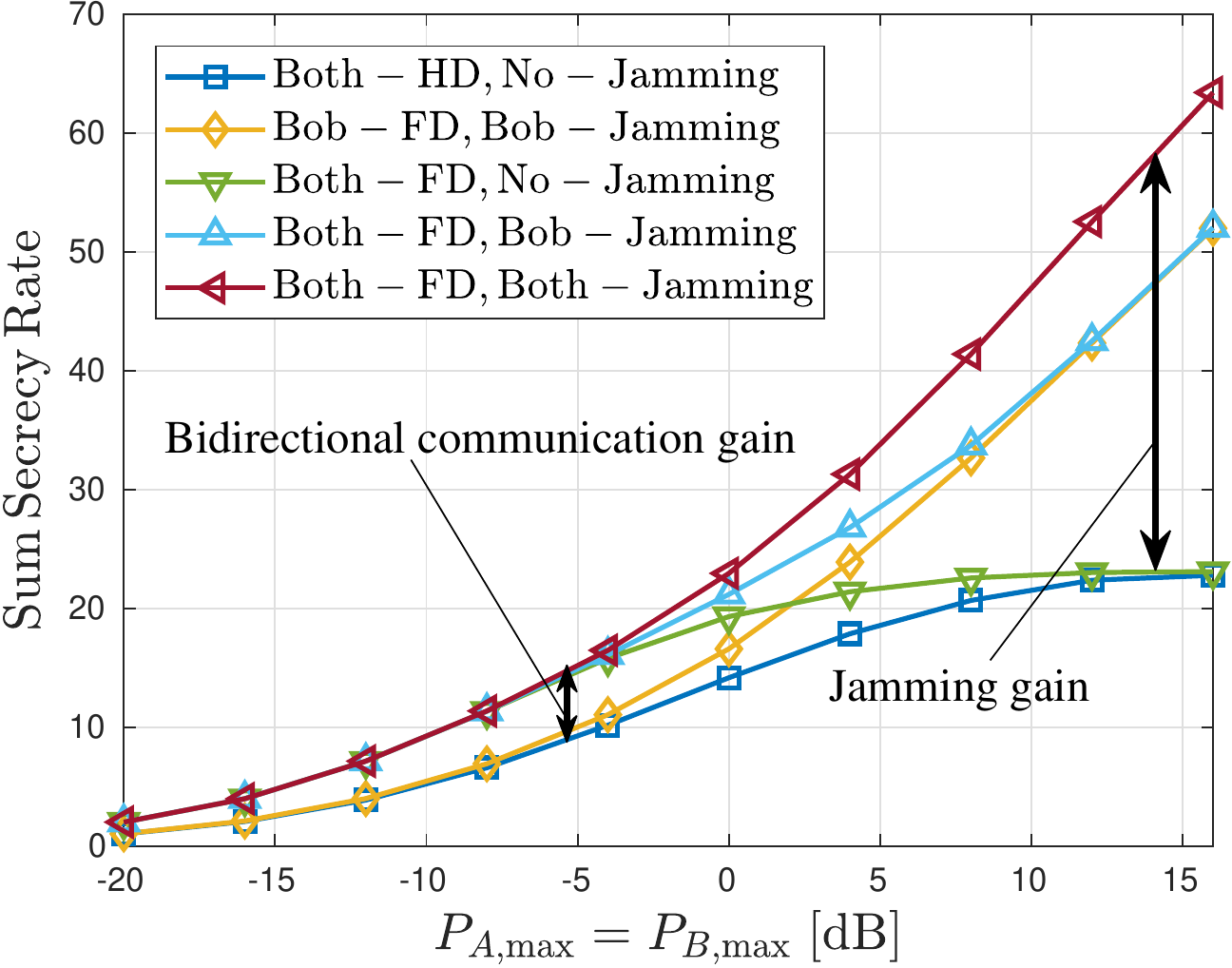}
    \caption{{{Performance of bidirectional secure communication related to maximum transmit power per node. }}} \label{fig:bothFD_power}
    \end{center} 
\end{figure}

\begin{figure}[!t] 
    \begin{center}
        \includegraphics[angle=0,width=0.89\columnwidth]{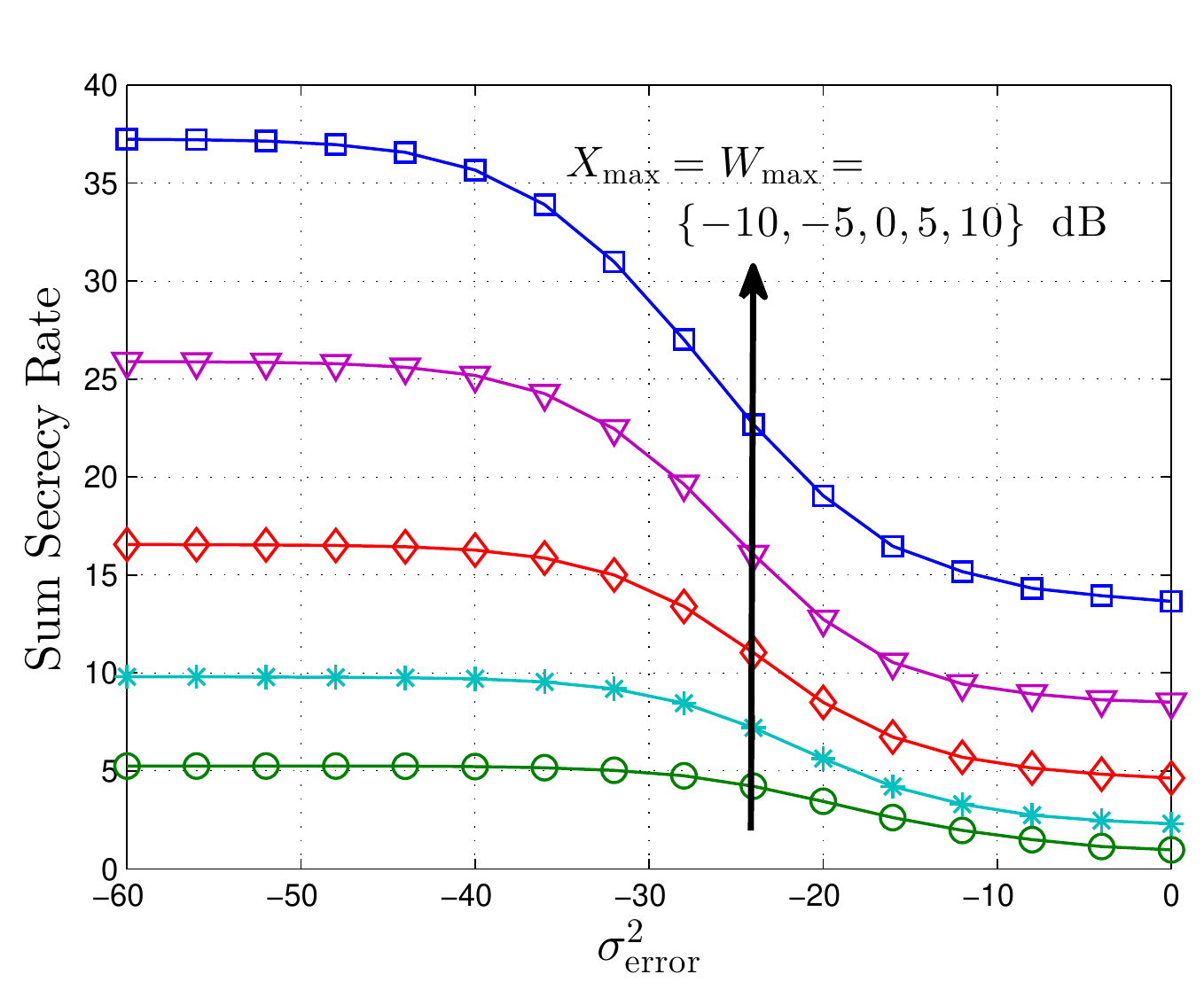}
    \caption{{{System sensitivity of proposed design with respect to the CSI error. }}} \label{fig:imperfectCSI}
    \end{center} 
\end{figure}

In this section we numerically evaluate the resulting sum secrecy rate of the defined system, comparing different system aspects and design strategies. In this respect, we assume that the channels $\ma{H}_{X}^{(n)}$ are following an uncorrelated Rayleigh distribution, with variance $\eta_X$ for each element, where $X\in\{ab, ba, ae, be\}$. Furthermore, $\ma{H}_{bb}^{(n)} \sim \mathcal{CN}\left( \sqrt{\frac{K_R}{1+K_R}} \ma{H}_0 , \frac{1}{1+K_R} \ma{I}_{M_{br}} \otimes \ma{I}_{M_{bt}}  \right)$, following \cite{6353396}, where $\ma{H}_0$ is a matrix with all elements equal to $1$ and $K_R$ is the Rician coefficient. The statistics of the self-interference channel on Alice, i.e., $\ma{H}^{(n)}_{aa}$, is defined similarly. The resulting system performance is then averaged over 200 channel realizations. Unless otherwise is stated we use the following values to define our default setup: $M_a = M_{at}= M_{ar} = 4$, $M_{b}=M_{bt}= M_{br} = 4$, $M_e =4$, $|\mathcal{N}|=4$, $K_R =10$, $X_{\text{max}}= W_{\text{max}} = P_{A,\text{max}} = P_{B,\text{max}} = 0\text{dB}$, $\kappa=\kappa^{(n)} = -30\text{dB}$, $\beta = \beta^{(n)} = -30\text{dB}$, $N_a= N_a^{(n)}= -30\text{dB}$, $N_b= N_b^{(n)}= -30\text{dB}$, $N_e= N_e^{(n)}= -30\text{dB}$, $\eta_{ab}=\eta_{ba}=\eta_{ae}=\eta_{be}=-20\text{dB}$, $\hat{\ma{D}}_{\mathcal{X}\mathcal{X}}=\hat{\ma{D}}^{(n)}_{\mathcal{X}\mathcal{X}} = \ma{D}^{(n)}_{\mathcal{X}\mathcal{X}}{\ma{D}^{(n)}_{\mathcal{X}\mathcal{X}}}^H = \ma{0}_{M_{\mathcal{X}r}\times M_{\mathcal{X}r}}, \mathcal{X} \in \{a,b\}$. \par

\subsection{Algorithm analysis}\label{algorithmAnalysis}
In this part the average convergence behavior and the  computational complexity of the proposed algorithm are studied. Moreover the impact of the choice of initialization points are evaluated.

In Fig.~\ref{fig:convergence} the average convergence behavior of the proposed iterative method is depicted. As it is observed, the convergence is obtained within 10-20 optimization iterations, which indicates the efficiency of the proposed iterative algorithm in terms of the required computational effort.

In Fig.~\ref{fig:cpu} the average required computation time for single directional system (`One-direction') and bidirectional system (`Bi-direction') related to the equipped transmit/receive antenna number of all nodes is depicted\footnote{The reported computation time is obtained using an Intel Core i7 4790S processor with the clock rate of 3.2 GHz and 16 GB of random access memory (RAM). The software platform is CVX \cite{cvx,gb08} with MATLAB 2014a on a 64-bit operating system.}. It is observed that a higher antenna array size results in a higher required computational complexity, associated with slower convergence and larger problem dimensions. Moreover, due to the additional problem complexity, the bidirectional communication system with joint FD-enabled jamming results in a higher computation time.

In Fig.~\ref{fig:ini} the impact of the initialization method is depicted. `Uniform-covariance', `Optimal-spatial', `Random-max', `Random-avg' represent the uniform covariance with equal power allocation initialization, optimal spatial beam initialization, maximal and average value of random initialization, respectively. It is observed that under high SIC level, i.e., low $\kappa,\beta$, the uniform covariance with equal power allocation initialization method reaches close to the benchmark performance\footnote{The benchmark performance is obtained by repeating the algorithm with 20 random initializations and choosing the highest obtained sum secrecy rate.}. Nevertheless, the optimal spatial beam initialization method results in a worse sum secrecy rate, however, within 0.36\% of the relative difference. Conversely, under low SIC level, i.e., high $\kappa, \beta$, the algorithms associated with uniform covariance with equal power allocation initialization converge to a local optimal point with a very small iteration number, which results in a relatively large difference margin (6-7\%) compared to the benchmark. Nevertheless, the optimal spatial beam initialization method has a close performance compared to the benchmark in this case.         

\subsection{Performance comparison}
In this part the performance of the proposed FD-enabled system and bidirectional communication system are evaluated under different system conditions. The performance between the FD-enabled setup and HD setup are also compared.  
\subsubsection{FD Bob jamming}
In Figs.~\ref{fig:wmax}-\ref{fig:Antenna_Eve2} the resulting sum secrecy rate is depicted considering different design strategies. In this respect, `Optimal-FD' represents the proposed design in Section~\ref{sec:SSRM}, supporting a FD jamming receiver. `Optimal-HD' represents a similar setup with a HD Bob. `Equal-FD' (`Equal-HD') is the setup with no optimization, i.e., a uniform power and beam allocation in all subcarriers for a system with an FD (HD) Bob. Furthermore, `Equal-X, Optimal-W', represents the case with equal power and beam allocation over all subcarriers for Alice together with an optimal design of the jammer. Conversely, `Equal-W, Optimal-X' represents the case with equal power and beam allocation for Bob together with an optimal design for Alice. 

In Fig.~\ref{fig:wmax} the impact of the maximum jamming power from Bob is depicted. A considerable gain is observed in this respect for a system with optimized jamming. Nevertheless, such gain is limited as $W_{\text{max}}$ increases. This stems from the fact that while jamming results in the degradation of Alice-Eve channel, the secrecy rate is bounded due to the limited Alice-Bob channel capacity. Moreover, it is observed that such jamming gain is only obtained by applying an optimally designed jamming transmit strategy. This emphasizes the impact of residual SI on the Alice-Bob communication, which should be controlled via jamming optimization. As a result, for a system with no jamming optimization, a high $W_{\text{max}}$ results in a significantly lower secrecy rate due to the impact of residual SI.

In Fig.~\ref{fig:xmax} the impact of the maximum transmit power from Alice on the obtained sum secrecy rate is depicted. Is is observed that as $X_{\text{max}}$ increases, the system obtains a higher sum secrecy rate. The performance gain, due to the optimization of transmit strategies, and due to the jamming capability at Bob is observed.

In Fig.~\ref{fig:kappa} the impact of the transceiver dynamic range is depicted. It is observed that as $\kappa$($=\beta$) increases, the jamming gain decreases due to the impact of residual SI. In this respect, for a transceiver with large values of $\kappa=\beta$ the jamming is turned-off for an optimally-designed system. On the other hand, a high $\kappa$ results in a sever degradation of the system performance due to the impact of residual SI, if the jamming strategy is not optimally controlled.  

In Fig.~\ref{fig:Antenna_Bob2} and Fig.~\ref{fig:Antenna_Eve2} the obtained sum secrecy rate is evaluated for different number of antennas at Bob and Eve. As expected, a more powerful Eve, i.e., higher $M_e$, results in a reduced system secrecy. In this respect, the gain of FD jamming becomes clear in combating the increasing quality of Alice-Eve channel. On the other hand, it is observed that the resulting secrecy improves as $M_{b}$ increases. In particular, the gain of FD jamming becomes significant as $M_{b}$ increases, as the jamming beam can be directed to Eve more efficiently. Furthermore, for smaller number of antennas at Bob, the optimization at Alice gains significance. This is perceivable, as a smaller $M_b$ results in a smaller design freedom at Bob, and also a weaker Alice-Bob channel.   

\subsubsection{Secure bidirectional communication}
In Fig.~\ref{fig:bothFD_power} a bidirectional secure communication system is studied. Three scenarios are considered regarding the jamming capability. Specifically, `Both-FD, No-Jamming', `Both-FD, Bob-Jamming' and `Both-FD, Both-Jamming' represent the bidirectional communication system with FD operation at Alice and Bob which is without jamming capability, with jamming capability only at Bob and with jamming capability at both Alice and Bob, respectively. Moreover two scenarios of the single direction communication system are also evaluated. Specifically, `Both-HD, No-Jamming' represents the system with HD operation at Alice and Bob without jamming capability and `Bob-FD, Bob-Jamming' represents the system with a HD Alice and a FD Bob as a jamming receiver.

It is observed that the bidirectional communication system leads a considerable enhancement of sum secrecy rate in a wide range of $P_{A,\text{max}},P_{B,\text{max}}$. Moreover, the jamming impact is more significant with large available power. From the results of `Both-FD, Both-Jamming' and `Both-FD, Bob-Jamming' it is also observed that the bidirectional jamming leads a higher sum secrecy rate in the studied bidirectional system, due to the reused jamming power for both communication directions. 

\subsection{Sensitivity to CSI error}
In Fig.~\ref{fig:imperfectCSI} the sensitivity of the proposed design, in terms of the resulting sum system secrecy rate, is observed with respect to the CSI error. The CSI error is modeled as $\tilde{H}^{(n)}_X = {H}^{(n)}_X  + {E}^{(n)}_X$, $X\in\{ab, ae, be\}$, where ${E}^{(n)}_X$ is modeled as a matrix with Gaussian i.i.d. elements with variance $\sigma_{\text{error}}^2$. It is observed that as the CSI accuracy decreases, the performance of the proposed design decreases. Nevertheless, the performance converges to its minimum level as $\sigma_{\text{error}}^2$ increases, as a high $\sigma_{\text{error}}^2$, is equivalent of having \textit{no knowledge} of the communication channels. Moreover, a system with a higher power level is more sensitive to CSI accuracy, compared to a system with a smaller $W_{\text{max}} = X_{\text{max}}$. This stems in the fact that as the transmit power decreases, the significance of the user noise increases and acts as the dominant factor in the system. In this respect, the impact of the CSI error becomes less significant, as noise acts as the dominant source of signal uncertainty.

\section {Conclusion} \label{sec:conclusion}
In this paper we have studied a joint power and beam optimization problem for a multi-carrier and MIMO wiretap channel in both single directional and bidirectional communication systems, where FD transceivers are capable of jamming. It is observed that for a system with an adequately high SI cancellation capability, an optimal jamming strategy results in a significant improvement of the sum secrecy capacity. In particular, in a frequency selective setup, the frequency diversity in different subcarriers can be opportunistically used, both regarding the jamming and the desired information link, to jointly enhance the resulting secrecy capacity. Nevertheless, it is observed that a jamming strategy with no power and/or beam optimization may lead to a reduced system secrecy, particularly as the SI cancellation capability decreases. Moreover, a promising sum secrecy gain is obtained from an FD bidirectional communication, where jamming power can be reused to improve security for both directions.  

{
 \bibliographystyle{IEEEtran}  
\bibliography{collection}

}    


\end{document}